\documentclass[sort&compress]{elsarticle} 
\usepackage[text={16cm,23cm},centering]{geometry}
\usepackage{mhsetup}
\usepackage{mathtools}  
\usepackage{amssymb}
\numberwithin{equation}{section}
\usepackage{amsthm}
\usepackage[mathscr]{eucal}  
\usepackage{extarrows}
\usepackage[defaultcolor=purple]{changes}

\usepackage{cases,paralist}
\usepackage[all,poly,knot]{xy}

\usepackage[right,mathlines]{lineno}  

\usepackage{xcolor}
\usepackage[colorlinks=true,allcolors=teal,bookmarksnumbered=true]{hyperref} 

\usepackage[shortlabels]{enumitem}  
\setlist[enumerate]{label=\upshape(\arabic*),leftmargin=*}

\usepackage{doi}
\usepackage[nameinlink,capitalize]{cleveref} 
\crefname{subsection}{Subsection}{Subsections}
\crefname{theorem}{Theorem}{Theorems}
\crefname{lemma}{Lemma}{Lemmas}
\crefname{corollary}{Corollary}{Corollaries}
\crefformat{equation}{#2\textup{(#1)}#3}
\crefmultiformat{equation}{(#2#1#3)}{ and~(#2#1#3)}{, (#2#1#3)}{ and~(#2#1#3)}

\theoremstyle{plain}
\newtheorem{theorem}{Theorem}[section]
\newtheorem{lemma}[theorem]{Lemma}

\newtheorem{corollary}[theorem]{Corollary}

\theoremstyle{definition}
\newtheorem{definition}{Definition}[section]
\newtheorem{example}{Example}[section]

\theoremstyle{remark}

\usepackage{xspace}  

\newcommand{\ifa}{if and only if\xspace}
\newcommand\qtq[1]{\quad\text{#1}\quad\xspace}

\newcommand{\n}{\mathbb{N}}
\newcommand{\z}{\mathbb{Z}}

\newcommand{\f}{\mathbb{F}}

\newcommand{\ord}{\mathrm{ord}}
\newcommand{\lcm}{\mathrm{lcm}}

\newcommand{\onetoone}{$1$-to-$1$\xspace}
\newcommand{\twotoone}{$2$-to-$1$\xspace}

\newcommand\mtoone{$m$-to-$1$\xspace}

\newcommand\mfq{$m$-to-$1$ on~$\mathbb{F}_{q}$\xspace}
\newcommand\mfqstar{$m$-to-$1$ on~$\mathbb{F}_{q}^{*}$\xspace}

\newcommand\mfield[2]{$#1$-to-$1$ on~$#2$\xspace}
\newcommand\mset[2]{$#1$-to-$1$ on~$#2$\xspace}

\newcommand{\fq}{\mathbb{F}_{q}}
\newcommand{\fqx}{\mathbb{F}_{q}[x]}
\newcommand{\fqstar}{\mathbb{F}_{q}^{*}}

\newcommand{\fqtwo}{\mathbb{F}_{q^2}}
\newcommand{\fqtwox}{\mathbb{F}_{q^2}[x]}
\newcommand{\fqtwostar}{\mathbb{F}_{q^2}^{*}}

\newcommand{\fqn}{\mathbb{F}_{q^n}}
\newcommand{\fqnx}{\mathbb{F}_{q^n}[x]}
\newcommand{\fqnstar}{\mathbb{F}_{q^n}^{*}}

\journal{XXX}
\date{}
\begin{document}
\begin{frontmatter}
\title{On many-to-one property of generalized cyclotomic mappings}
\tnotetext[t1]{
The Magma codes involved in this paper can be downloaded from \href{http://dx.doi.org/10.13140/RG.2.2.20938.89287}{here}.
This work was partially supported by 
the Natural Science Foundation of Shandong (No.\ ZR2021MA061),
the National Natural Science Foundation of China (Nos.\ 12461103, 12471492, 62072161),
the Innovation Group Project of the Natural Science Foundation of Hubei (No.\ 2023AFA021), 
and NSERC of Canada (RGPIN-2023-04673). 
}

\author[QF]{Yanbin Zheng}
\ead{zheng@qfnu.edu.cn}

\author[QF]{Yang Zhang}
\ead{2224959114@qq.com}

\author[XJ]{Zhengbang Zha}
\ead{zhazhengbang@163.com}

\author[WH]{Xiangyong Zeng}
\ead{xzeng@hubu.edu.cn}

\author[OT]{Qiang Wang}
\ead{wang@math.carleton.ca}

\address[QF]{School of Mathematical Sciences, Qufu Normal University, Qufu 273165, China}

\address[XJ]{College of Mathematics and System Science, Xinjiang University, Urumqi 830017, China}

\address[WH]{Hubei Key Laboratory of Applied Mathematics,
Faculty of Mathematics and Statistics,\\ Hubei University, Wuhan 430062, China}

\address[OT]{School of Mathematics and Statistics, 
Carleton University, 1125 Colonel By Drive,\\ Ottawa, ON K1S 5B6, Canada}

\begin{abstract}
The generalized cyclotomic mappings over finite fields $\mathbb{F}_{q}$ are those mappings which induce monomial functions on all cosets of an index $\ell$ subgroup $C_0$ of the multiplicative group $\mathbb{F}_{q}^{*}$. Previous research has focused on the one-to-one property, the functional graphs, and their applications in constructing linear codes and bent functions. In this paper, we devote to study the many-to-one property of these mappings. We completely characterize many-to-one generalized cyclotomic mappings for $1 \le \ell \le 3$. Moreover, we completely classify $2$-to-$1$ generalized cyclotomic mappings for any divisor $\ell$ of $q-1$. In addition, we construct several classes of many-to-one binomials and trinomials of the form $x^r h(x^{q-1})$ on $\mathbb{F}_{q^2}$, where $h(x)^{q-1}$ induces monomial functions on the cosets of a subgroup of $U_{q+1}$. 
\end{abstract}


\end{frontmatter}
 
\section{Introduction}
For a prime power $q$, 
let $\fq$ denote the finite field~with~$q$ elements,
$\fqstar =\fq\setminus\{0\}$, 
and $\fqx$ be the ring of polynomials over~$\fq$.
Let $\xi$ be a primitive element of $\fq$ 
and $q-1 = \ell s$ for some $\ell, s \in \n$. 
Let $C_0 =  \{\xi^{k\ell} : 0 \le k < s \}$ 
be the set of all $s$-th roots of unity in  $\fq^*$.
Then $C_0$ is a subgroup of $\fqstar$ of index~$\ell$.
The elements of the factor group $\fqstar/C_0$ 
are the cyclotomic cosets
\begin{equation}\label{Ci}
    C_i = \xi^{i} C_0 = \{\xi^{k\ell+i} : 0 \le k < s\}, 
    \quad  0 \le i \le \ell - 1,
\end{equation}
which form a partition of $\fqstar$. 
For $a_{0}$, $a_{1}$, \ldots,  $a_{\ell-1}\in \fq$ and 
$r_{0}$, $r_{1}$, \ldots,  $r_{\ell-1} \in \n$, 
the generalized cyclotomic mappings from $\fq$ 
to itself is defined in~\cite{Wang-cyc} by
\begin{equation}\label{map_aixri}
  f(x) = 
  \begin{cases}
  0             & \text{if $x =  0$,} \\
  a_i x^{r_i}   & \text{if $x \in C_i$, 
                  $0 \le i \le \ell - 1$.}
  \end{cases}   
\end{equation}
(Cyclotomic mappings were introduced in~\cite{Nied-cyc} 
when $r_{0} = r_{1} = \dotsm = r_{\ell-1} =  1$ 
and in~\cite{Wang07} for $r_{0} = r_{1} = \dotsm = r_{\ell-1} >1$. 
Further information can be found in~\cite[Section~8.1.5]{HFF}.) 
For some $0 \le i, t \le \ell-1$ and $x \in C_t$, 
we have $x^s = \omega^t$ where $ \omega = \xi^s$, and so
\[
  \sum_{j=0}^{\ell-1}\Big(\frac{x^s}{ \omega^i}\Big)^j 
  = \sum_{j=0}^{\ell-1} \omega^{(t-i)j}
  = \begin{cases}
    0      & \text{if $t \ne i$,} \\
    \ell   & \text{if $t = i$.}
 \end{cases}
\]
Hence $f(x)$ in~\cref{map_aixri} is rewritten in~\cite{Wang-cyc} as
\begin{equation}\label{poly_aixri}
  f(x) 
  = \frac{1}{\ell}\sum_{i=0}^{\ell-1} a_{i} x^{r_i}
    \sum_{j=0}^{\ell-1} \Big(\frac{x^s}{ \omega^i}\Big)^j
  = \frac{1}{\ell}
    \sum_{i=0}^{\ell-1}\sum_{j=0}^{\ell-1}
    a_{i}  \omega^{- i j} x^{r_i + j s},
\end{equation}
in the sense of reduction modulo $x^q -x$. 
Theorem~2.2 in~\cite{Wang-cyc} gives several 
equivalent necessary and sufficient conditions 
for~$f(x)$ in~\cref{poly_aixri} permuting $\fq$.
An explicit expression of the inverse of~$f(x)$ on $\fq$ was found 
in \cite{ZhengPW2} by using a piecewise interpolation formula.
Later, based on the cyclotomic characteristics of~$f(x)$,
a shorter proof of the main result in \cite{ZhengPW2} 
and a necessary and sufficient condition for~$f(x)$  
to be an involution over $\fq$ were given in \cite{Wang-cyc2}.

Besides the characterization and construction of permutation polynomials, 
different representations including the wreath product form 
of generalized cyclotomic mappings and 
their cycle structures are studied in \cite{BorsWang22}. 
Furthermore, the functional graphs of generalized cyclotomic mappings 
of $\fq$ were extensively studied in \cite{BorsPW23}.
The applications of these mappings in constructing binary linear codes 
with few weights or bent functions can be found in 
\cite{Ding072274,Fang2023} and \cite{Xie2023} respectively. 
This motivates us to further study these generalized cyclotomic mappings. 

In 2019, Mesnager and Qu \cite{2to1-MesQ19} 
introduced the definition of \twotoone mappings 
and provided a systematic study of 
\twotoone mappings over finite fields. Later, 
the concept of \twotoone mappings was generalized in 
\cite{GAO211612, BartGT22194, nto1-NiuLQL23}.
Very recently, Zheng et al. \cite{Zhengmto1} introduced
the definition of many-to-one (\mtoone for short) 
mappings between two finite sets, 
which unifies and generalizes the definitions in
\cite{2to1-MesQ19, GAO211612, BartGT22194, nto1-NiuLQL23}.
The main purpose of this paper is to study the \mtoone property 
of these generalized cyclotomic mapping $f(x)$ in~\cref{map_aixri}
under the definition of \mtoone mappings in \cite{Zhengmto1}.

The rest of the paper is organized as follows.
After some preliminaries are given in \cref{sec_pre}, 
the \mtoone property of $f(x)$ in~\cref{map_aixri}
is completely characterized for index $\ell = 2$ or~$3$ 
in \cref{sec_ell=2,sec_ell=3}, respectively. 
\cref{sec_ell>=3} determines the \twotoone property 
of $f(x)$ for any divisor $\ell$ of $q-1$.
Then the \mtoone property of $f(x)$ is also completely 
characterized if all $(r_i, s)$'s are equal.
In \cref{sec_Realization}, several classes of \mtoone mappings 
are explicitly obtained by using polynomials of special forms 
to represent the constants $a_i$'s. 
Finally, by using the \mtoone property of $x^r h(x)^{q-1}$ on $U_{q+1}$,
several classes of \mtoone binomials and trinomials 
of the form $x^r h(x^{q-1})$ on $\fqtwo$ are given, 
where $h(x)^{q-1}$ induces monomial functions 
on the cosets of a subgroup of $U_{q+1}$.

In this paper, we shall use the following standard notation.
The letter $\z$ will denote the set of all integers,
$\n$ the set of all positive integers,
$\# S$ the cardinality of a finite set $S$,
and $\varnothing$ the empty set containing no elements. 
The greatest common divisor of two integers 
$a$ and $b$ is written as $(a, b)$, 
and $\lcm (a, b) = a b / (a, b)$. 
Denote $a \bmod n$ as the smallest non-negative 
remainder obtained when~$a$ is divided by~$n$; 
that is, $\mathrm{mod}~n$ is a function 
from $\z$ to $\{0, 1, 2, \ldots, n-1\}$.
Moreover, we define 
\[
\phi(i, n) = (i r_i + \log_{\xi} a_i) \bmod{n},
\]
where $i$, $r_i$, $a_i$ are as in~\cref{map_aixri} 
and $\xi$ is a primitive element of $\fq$.

\section{Preliminaries} \label{sec_pre}

We first recall the definition of many-to-one mapping  
introduced in \cite{Zhengmto1}.

\begin{definition}[\cite{Zhengmto1}]\label{defn:mto1}
  Let $A$ be a finite set and 
  $m \in \z$ with $1 \le m \le \# A$. 
  Write $\# A = k m + r$, 
  where $k$, $r \in \z$ with $0 \le r < m$. 
  Let~$f$ be a mapping from~$A$ to another finite set~$B$. 
  Then~$f$ is called many-to-one, 
  or \mtoone for short, on~$A$ if 
  there are~$k$ distinct elements in~$B$ such that 
  each element has exactly~$m$ preimages in~$A$ under~$f$.
  The remaining~$r$ elements in~$A$ are called the
  \textit{exceptional elements} of~$f$ on~$A$,
  and the set of these~$r$ exceptional elements 
  is called the \textit{exceptional set} 
  of~$f$ on~$A$ and denoted by $E_{f}(A)$. 
  In particular, $E_{f}(A) = \varnothing$ 
  \ifa $r = 0$, i.e., $m \mid \# A$.
\end{definition}

\begin{definition}[\cite{Zhengmto1}]
A polynomial $f(x) \in \fqx$ is called many-to-one, 
  or \mtoone for short, on $\fq$ if the mapping 
  $f \colon c \mapsto f(c)$ from $\fq$ to itself is \mset{m}{\fq}.
\end{definition}

A direct computation gives the following example.

\begin{example}
Let $f(x) = x^3 + x$. Then $f$ maps $0,1,2,3,4$ 
to $0,2,0,0,3$ in~$\f_5$, respectively. 
Thus $f$ is \mfield{3}{\f_5} and the 
exceptional set $E_{f}(\f_5) = \{1, 4\}$.
\end{example}

Since $C_0$ is a cyclic subgroup of $\fqstar$ 
and $C_i = \xi^{i} C_0$, we obtain the next result.


\begin{lemma}\label{lem:fi(ci)}
Let $q-1 = \ell s$ for some $\ell, s \in \n$.
Let $a_i \in \fqstar$, $r_i \in \n$,  and $d_i = (r_i, s)$, 
where $0 \le i \le \ell - 1$.
Then $f(x)$ in~\cref{map_aixri} is $d_i$-to-$1$ 
from $C_i$ to $C_{\phi(i, \ell)}$
and the exceptional set $E_{f}(C_i) = \varnothing$.
\end{lemma}
\begin{proof}
By \cref{map_aixri}, 
$f(x) = f_i(x) \coloneq a_i x^{r_i}$ 
for any $x \in C_i$.
Since $(r_i/d_i, s/d_i) = 1$, we have
$(r_i/d_i) x$ permutes $\z_{s/d_i}$,
and so 
\begin{align*}
f_i(C_i) 
& = a_i \{ \xi^{\ell + i}, \xi^{2 \ell + i}, \ldots, \xi^{(s/d_i) \ell + i} \}^{r_i}
  = a_i \xi^{i r_i} \{ \xi^{\ell}, \xi^{2 \ell}, \ldots, \xi^{(s/d_i) \ell} \}^{r_i}\\
& = a_i \xi^{i r_i} \{ \omega_i, \omega_i^2, \ldots, \omega_i^{s/d_i}\}^{r_i/d_i}
  = a_i \xi^{i r_i} \{ \omega_i, \omega_i^2, \ldots, \omega_i^{s/d_i}\},
\end{align*}
where the symbol $A^{r_i}$ denotes the set $\{a^{r_i} : a \in A\}$,
$\omega_i = \xi^{\ell d_i}$, 
and $\ord (\omega_i) = (q-1)/(\ell d_i) = s/d_i$.
Thus $f_i$ is $d_i$-to-$1$ from $C_i$ to $C_{\phi(i, \ell)}$.
\end{proof}

In particular, if $\ell = 1$, then $f(x) = a_0 x^{r_0}$.
Thus $f$ is \mfield{(r_0, q-1)}{\fqstar} and $E_{f}(\fqstar) = \varnothing$. 
We next recall a well-known result of linear Diophantine equations in two variables.

\begin{lemma}\label{lem:ax+by=c}
Let $a$, $b$, $c \in \z$ with $d = (a, b)$. The equation 
$a x + b y = c$ has integral solutions \ifa $d \mid c$.
If $x = x_0$, $y = y_0$ is a particular solution of the equation, 
then all integer solutions are given by 
\[
    x = x_0 + (b/d)\, t, \quad 
    y = y_0 - (b/d)\, t, \quad
    \text{$t \in \z$.}
\]
\end{lemma}

\begin{lemma}\label{lem:AB}
Let $a$, $b$, $c$, $n \in \n$ satisfy $a \mid n$ and $b \mid n$. Let $d = (a, b)$,
\[
  A = \{(a x) \bmod{n} : x \in \z \}
  \quad \text{and} \quad 
  B = \{(b y + c) \bmod{n} : y \in \z\}.
\]
Then the following statements hold:
\begin{enumerate}[\upshape(1)]
\item $A \cap B = \varnothing$ \ifa $d \nmid c$;
\item $B \subseteq A$ \ifa $a \mid b$ and $a \mid c$; 
\item $A \cap B \ne \varnothing$ and $B \not \subseteq A$ \ifa  
    $a \nmid b$ and $d \mid c$. Moreover, 
    \[
      A \cap B 
      = \{(a x_0 + t\, \lcm(a,b) ) \bmod{n} : 0 \le t < n/\lcm(a, b)\},
    \]
    where $x_0 = \bar{a} c / d$ and 
    $\bar{a}$ is the inverse of $a/d$ modulo $b/d$.
\end{enumerate}
\end{lemma}

\begin{proof}
Clearly, $A \cap B = \varnothing$ \ifa  $a x  - b y = c + k n$ 
has no integral solutions for any $k \in \z$,
which is equivalent to $d \nmid c$ by \cref{lem:ax+by=c} and $d \mid n$. 
Thus $A \cap B \ne \varnothing$ \ifa $d \mid c$.
Let $x_0 = \bar{a} c / d$ and $y_0 = c (a \bar{a} - d) / (b d)$.
It is easy to verify that $a x_0 - b y_0 = c$.
Then by \cref{lem:ax+by=c},
\begin{align*}
A \cap B
& = \{ a (x_0 + t (b / d) ) \bmod{n} : t \in \z \} \\
& = \{(a x_0 + t\, \lcm(a,b) ) \bmod{n} : 0 \le t < n/\lcm(a, b)\}.
\end{align*}
In particular, $B \subseteq A$ \ifa $\#(A \cap B) = \# B = n/b$,
i.e., $a \mid b$ and $a\mid c$.
\end{proof}

Now we use \cref{lem:fi(ci),lem:AB} to characterize 
the relationships between $f_i(C_i)$ and $f_j(C_j)$.

\begin{theorem}\label{thm:fiCifjCj}
Let $q-1 = \ell s$ for some $\ell, s \in \n$.
Let $f_i(x) = a_i x^{r_i}$ 
and $d_i = (r_i, s)$, where $a_i \in \fqstar$, $r_i \in \n$, 
and $0 \le i \le \ell - 1$.
For any $0 \le i, j \le \ell - 1$, let  $d = (d_i, d_j)$ and $\bar{d} = \lcm (d_i, d_j)$. Then 
\begin{enumerate}[\upshape(1)]
    \item $f_i(C_i) \cap f_j(C_j) = \varnothing$ \ifa
        $\phi(i, \ell d) \ne \phi(j, \ell d)$;
    \item $f_i(C_i) \subseteq f_j(C_j)$ \ifa 
$d_j \mid d_i$ and $\phi(i, \ell d_j) = \phi(j, \ell d_j)$;
    \item $f_i(C_i) \cap f_j(C_j) \ne \varnothing$ and $f_i(C_i) \not \subseteq f_j(C_j)$ \ifa  
        $d_j \nmid d_i$ and $\phi(i, \ell d) = \phi(j, \ell d)$.
        Moreover,
        \begin{equation}\label{eq:fiCicapf_jCj}
            f_i(C_i) \cap f_j(C_j) 
            = a_j \xi^{j r_j + \ell d_j x_0} 
            \{\xi^{\ell   \bar{d} t} : 0 \le t < s/\bar{d} \},
        \end{equation}
        where $x_0 = \bar{a} c /(\ell d)$, $c = (i r_i - j r_j) + \log_{\xi} {(a_i/a_j)}$,
        and $\bar{a}$ is the inverse of $d_j/d$ modulo $d_i/d$.
\end{enumerate}
\end{theorem}

\begin{proof}
From the proof of \cref{lem:fi(ci)}, we have 
\[
  f_i(C_i) = a_i \xi^{i r_i} \{ \omega_i, \omega_i^2, \ldots, \omega_i^{s/d_i}\},  
\]
where $\omega_i = \xi^{\ell d_i}$ and $\ord (\omega_i) = s/d_i$.
Then $f_i(C_i) \cap f_j(C_j) = \varnothing$ \ifa
\begin{equation}\label{eq:ficapfj}
(a_i/a_j) \xi^{i r_i - j r_j} 
\{\omega_i, \omega_i^2, \ldots, \omega_i^{s/d_i}\} \cap \{ \omega_j, \omega_j^2, \ldots, \omega_j^{s/d_j}\}
= \varnothing.
\end{equation}
Let $c = (i r_i - j r_j) + \log_{\xi} {(a_i/a_j)}$, 
\[
  A = \{(\ell d_j x) \bmod{(q-1)} : x \in \z \}
  \quad \text{and} \quad 
  B = \{(\ell d_i + c) \bmod{(q-1)} : y \in \z\}.
\]
Then \cref{eq:ficapfj} is equivalent to 
$A \cap B = \varnothing$. 
That is, $\ell d \nmid c$ 
by \cref{lem:AB}, i.e.,
 $\phi(i, \ell d) \ne \phi(j, \ell d)$.
Similarly, $f_i(C_i) \subseteq f_j(C_j)$ \ifa $B \subseteq A$, 
or equivalently, $d_j \mid d_i$ and $\phi(i, \ell d_j) = \phi(j, \ell d_j)$. 

By \cref{lem:AB}, if $A \cap B \neq \varnothing$, then 
\[
A \cap B 
      = \{(\ell d_j x_0 + \ell \bar{d} t ) \bmod{(q-1)} : 0 \le t < s / \bar{d}\},
\]
where $x_0 = \bar{a} c / d$
and $\bar{a}$ is the inverse of $d_j/d$ modulo $d_i/d$.
Hence
\[
(a_i/a_j) \xi^{i r_i - j r_j} 
\{\omega_i, \omega_i^2, \ldots, \omega_i^{s/d_i}\} \cap \{ \omega_j, \omega_j^2, \ldots, \omega_j^{s/d_j}\}
= \xi^{\ell d_j x_0} \{\xi^{\ell \bar{d} t} : 0 \le t < s / \bar{d}\},
\]
and thus \cref{eq:fiCicapf_jCj} holds. 
\end{proof}

\cref{thm:fiCifjCj} completely characterizes the relationships between 
$f_i(C_i)$ and $f_j(C_j)$.
In particular,  
it is reduced to the following form when $d_i = d_j$.

\begin{corollary}\label{f_i(C_i)_di=dj}
With the notation and the hypotheses of \cref{thm:fiCifjCj},
assume $d_i = d_j$ for some $0 \le i, j \le \ell - 1$.
Then $f_i(C_i) = f_j(C_j)$ if $\phi(i, \ell d) = \phi(j, \ell d)$,
and $f_i(C_i) \cap f_j(C_j) = \varnothing$ 
if $\phi(i, \ell d) \ne \phi(j, \ell d)$.
\end{corollary}



We also need the following results in the sequel.

\begin{lemma}\label{lem:di+dj}
Let $q-1 = \ell s$ for some $\ell, s \in \n$.
Let $f_i(x) = a_i x^{r_i}$ and $d_i = (r_i, s)$,
where  $a_i \in \fqstar$, $r_i \in \n$, 
and $0 \le i \le \ell - 1$. 
Assume $d_i \le d_j$ for some $0 \le i, j \le \ell - 1$, 
\begin{equation}\label{eq:joint}
f_i(C_i) \cap f_j(C_j) \ne \varnothing \qtq{and} 
f_j(C_j) \not\subseteq f_i(C_i).
\end{equation}
Then $f$ in~\cref{map_aixri} is not \mfield{(d_i + d_j)}{C_i \cup C_j}.
\end{lemma}

\begin{proof}
Assume $f$ is \mset{(d_i + d_j)}{C_i \cup C_j}. 
Denote $A = f_i(C_i) \setminus f_j(C_j)$
and $B = f_j(C_j) \setminus f_i(C_i)$.
Then for any $a \in A$ and $b \in B$, we have
$\# f_i^{-1}(a) = d_i$ and $\# f_j^{-1}(b) = d_j$.
Hence $A \cup B = f(E_{f}(C_i \cup C_j))$ by \cref{eq:joint}, 
and so $\# E_{f}(C_i \cup C_j) \ge d_i + d_j$,
contrary to the definition of \mtoone. 
\end{proof}

\begin{lemma}[{\cite[Lemma~4.2]{Zhengmto1}}]\label{fq=fqstar}
Assume $f \in \fqx$ has only the root $0$ in $\fq$.
Then $f$ is \mfield{1}{\fq} \ifa $f$ is \mfield{1}{\fqstar}.
If $m \ge 2$, then $f$ is \mfq \ifa $m \nmid q$ and 
$f$ is \mfqstar. 
\end{lemma}

This result reduces the problem whether~$f$ 
is an \mtoone mapping on~$\fq$ to that 
whether~$f$ is \mtoone on the multiplicative group~$\fqstar$.
Hence we need only find the conditions that~$f$ is \mfqstar.

\section{Two branches}\label{sec_ell=2}

In this section, we consider the case $f$ has two branches on $\fqstar$. 

\begin{theorem}\label{aixrimto1_2pieces}
Let $q$ be odd, $s = (q-1)/2$, $a_{0}, a_{1} \in \fqstar$, 
$r_0, r_1 \in \n$, and $d_i = (r_i, s)$. Then for $1 \le m < q$, 
\begin{equation}\label{eq_2pieces}
      f(x) = a_0 x^{r_0}(1 + x^s) + a_1 x^{r_1}(1 - x^s)
\end{equation}
is \mfqstar \ifa one of the following holds:
\begin{enumerate}[\upshape(1)]
  \item $m = d_0 = d_1$ and $\phi(0, 2m) \ne \phi(1, 2m)$;
  \item $m = d_0 + d_1$, $d \mid m$,
        $\phi(0, 2d) = \phi(1, 2d)$,
        and $s (m - 2d)/(m - d) < m$, 
        where $d = \min \{d_0, d_1\}$.        
\end{enumerate} 
\end{theorem}

\begin{proof}
Note that $f(x)$ is \mfqstar \ifa $F(x) := (1/2) f(x)$ is \mfqstar.
Clearly,
\begin{equation}\label{map_aixri2}
  F(x) = 
  \begin{cases}
  a_{0} x^{r_{0}}   & \text{for $x \in C_0$,} \\
  a_{1} x^{r_{1}}   & \text{for $x \in C_1$.}
  \end{cases}   
\end{equation}
Let $f_i (x) = a_{i}x^{r_i}$ for $i \in \{0, 1\}$. 
By \cref{lem:fi(ci)}, $f_i$ is \mset{d_i}{C_i}.
Obviously, when $d_0 > m$ or $d_1 > m$,
$F$ is not \mfqstar, and thus
$d_0, d_1  \in \{1, 2, \ldots, m\}$. 

When $d_0 = m$, we have $d_1 = d_0$. 
Indeed, if $d_1 \neq d_0$, then $b \in f_1 (C_1)$ has less 
than~$m$ preimages in $\fqstar$. 
Obviously, if $f_0 (C_0) \cap f_1 (C_1) = \varnothing$,
$F$ is not \mfqstar.
If $f_0 (C_0) \cap f_1 (C_1) \ne \varnothing$,
there exists $c \in f_0 (C_0) \cap f_1 (C_1)$
which has more than $m$ preimages in $\fqstar$,
and so $F$ is not \mfqstar.
Hence each $f_i$ is \mset{m}{C_i}.
Thus $F$ is \mfqstar \ifa $f_0 (C_0) \cap f_1 (C_1) = \varnothing$,
which is equivalent to $\phi(0, 2m) \ne \phi(1, 2m)$
by \cref{f_i(C_i)_di=dj}.

When $d_0 < m$, we prove that $d_0 + d_1 = m$.
If $d_1 < m - d_0$, then $F$ is obviously not \mfqstar.
If $m - d_0 < d_1 < m$, then $b$ has less than $m$ preimages 
in $\fqstar$ for all $b \in f_1 (C_1)$ when $f_0 (C_0) \cap f_1 (C_1) = \varnothing$,
and $c$ has more than $m$ preimages in $\fqstar$ 
for all $c \in f_0 (C_0) \cap f_1 (C_1)$ when $f_0 (C_0) \cap f_1 (C_1) \ne \varnothing$.
Thus $F$ is not \mfqstar.
Without loss of generality, we may assume 
$d := \min \{d_0, d_1\} = d_0$.
Since $f_i$ is \mset{d_i}{C_i} for $i \in \{0, 1\}$ 
and $d_1 = m - d \ge d_0$, we have 
\[
  \# f_1 (C_1) = s/(m - d) \le s/d = \# f_0 (C_0).
\]
Obviously, if $f_0 (C_0) \cap f_1 (C_1) = \varnothing$,
$F$ is not \mfqstar.
If $f_0(C_0) \cap f_1(C_1) \ne \varnothing$ and 
$f_1(C_1) \not\subseteq f_0(C_0)$, 
then $F$ is also not \mfqstar by \cref{lem:di+dj}.
Hence $F$ is \mfqstar \ifa 
$f_1 (C_1) \subseteq f_0 (C_0)$ and
\[
    \# E_{F}(\fqstar) = d (s/d - s/(m - d))
                      = s (m - 2 d)/(m - d)
                      < m.
\]
By \cref{thm:fiCifjCj}, 
$f_1 (C_1) \subseteq f_0 (C_0)$ \ifa
$d \mid m - d$ and $\phi(0, 2d) = \phi(1, 2d)$.
\end{proof}

\cref{aixrimto1_2pieces} completely characterizes the many-to-one 
property of the generalized cyclotomic mappings for $\ell = 2$.
When $m = 1$, \cref{aixrimto1_2pieces} is reduced to
\cite[Corollary~2.3]{Wang-cyc}.
Applying \cref{aixrimto1_2pieces} to $m = 2$ or~$3$ 
yields the following results.

\begin{corollary}\label{aixri2to1_2pieces}
Let the notation and the hypotheses 
be as in \cref{aixrimto1_2pieces}.
Then $f$ is \mfield{2}{\fqstar} \ifa one of the following holds:
\begin{enumerate}[\upshape(1)]
  \item \label{gcd(rs)=2} $d_0 = d_1 = 2$ and $\phi(0, 4) \ne \phi(1, 4)$;  
  \item \label{gcd(rs)=1} $d_0 = d_1 = 1$ and $\phi(0, 2) = \phi(1, 2)$.
\end{enumerate}  
\end{corollary}

If $r_1$ is even and $(a_1 / a_0)^{s} = -1$, 
then $\phi(0, 4) \ne \phi(1, 4)$.
Thus \cref{gcd(rs)=2} of \cref{aixri2to1_2pieces} generalizes 
\cite[Example~2.1]{Qin2469} where $a_0 = 1$ and $a_1^s = -1$.
Moreover, $(a_1 / a_0)^{s} = (-1)^{r_1}$ \ifa $\phi(0, 2) = \phi(1, 2)$.
Thus \cref{gcd(rs)=1} of \cref{aixri2to1_2pieces} generalizes
\cite[Example~2.2]{Qin2469} where $(a_1 / a_0)^{s} = -1$ and $r_1$ is odd, and 
\cite[Theorem~4.14]{nto1-NiuLQL23} where $q \equiv 3 \pmod{4}$ and $r_0 = r_1$.

\begin{corollary}\label{aixri3to1_2pieces}
Let the notation and the hypotheses 
be as in \cref{aixrimto1_2pieces}.
If $q \ge 13$, then $f$ is \mfield{3}{\fqstar} \ifa 
$d_0 = d_1 = 3$ and $\phi(0, 6) \ne \phi(1, 6)$. 
\end{corollary}

In \cref{aixri2to1_2pieces,aixri3to1_2pieces}, 
take $q = 13$, $s = 6$, and $\xi = 2$. 
Take $(a_0, a_1, r_0, r_1) = (1, -1, 2, 4)$.
Then $d_0 = d_1 = 2$, $\phi(0, 4) = 0$, and $\phi(1, 4) = 2$.
Take $(a_0, a_1, r_0, r_1) = (2, -1, 1, 5)$.
Then $d_0 = d_1 = 1$ and $\phi(0, 2) = \phi(1, 2) = 1$.
Take $(a_0, a_1, r_0, r_1) = (8, 7, 3, 3)$.
Then $d_0 = d_1 = 3$, $\phi(0, 6) = 3$, and $\phi(1, 6) = 2$.
Hence we obtain the next example.

\begin{example}
The polynomial $x^{10} + x^8 - x^4 + x^2$ is \mfield{2}{\f_{13}^{*}}, 
$x^{11} + 2 x^{7} - x^5 + 2 x$ is \mfield{2}{\f_{13}^{*}}, 
and $x^9 + 2 x^3$ is \mfield{3}{\f_{13}^{*}}.
\end{example}

\section{Three branches}
\label{sec_ell=3}

In this section, we consider the case~$f$ in 
\cref{map_aixri} has three branches on $\fqstar$.

\begin{theorem}\label{aixrimto1_3pieces}
Let $q-1 = 3 s$, $a_0, a_1, a_2 \in \fqstar$, 
$r_0, r_1, r_2 \in \n$, and 
\begin{equation}\label{eq_3pieces}
     f(x) = a_0 x^{r_0}(1 + x^s + x^{2s}) 
         + a_1 x^{r_1}(1 + \omega^2 x^s + \omega x^{2s})
        + a_2 x^{r_2}(1 + \omega x^s + \omega^2 x^{2s}),
\end{equation}
where $\omega$ is a primitive 
$3$-th root of unity over $\fq$.
Let $d_i \le d_j \le d_k$,
where $i$, $j$, $k \in \{0, 1, 2\}$ are distinct and $d_t = (r_t, s)$ for $t \in \{0, 1, 2\}$.
Then for $1 \le m < q$, $f$ is 
\mfqstar \ifa one of the following holds:
\begin{enumerate}[\upshape(1)]
    \item \label{item:L=3-1}
        $m = d_0 = d_1 = d_2$ 
        and $\phi(0, 3m)$, $\phi(1, 3m)$, and $\phi(2, 3m)$ are distinct;
    \item \label{item:L=3-2}
        $m = d_k = d_i + d_j$, $d_i \mid d_j$, 
        $\phi(i, 3 d_i) = \phi(j, 3 d_i) \ne \phi(k, 3 d_i)$,
        and $s (d_j - d_i) / d_j < m$;
    \item \label{item:L=3-3}
        $m = d_i + d_j$, $d_j = d_k$, $d_i \mid d_j$,   
        $\phi(i, 3 d_i) = \phi(j, 3 d_i) = \phi(k, 3 d_i)$, 
        $\phi(j, 3 d_j) \ne \phi(k, 3 d_j)$,
        and $s (d_j - 2 d_i) / d_j < m$; 
    \item \label{item:L=3-4}
        $m = 2s$, $d_0 = d_1 = d_2 = s$,
        and $\phi(i, 3s) = \phi(u, 3s) \ne \phi(v, 3s)$, 
        where $u$, $v \in \{j, k\}$ are distinct;  
    \item \label{item:L=3-5}
        $m = 2s$, $d_j = d_k = s$,
         $\phi(j, 3s) = \phi(k, 3s)$, and $\phi(i, 3 d_i) \ne \phi(j, 3 d_i)$;      
    \item \label{item:L=3-6}
        $m = d_0 + d_1 + d_2$, $d_i \mid d_k$, $d_j \mid d_k$,
        $\phi(i, 3 d_i) = \phi(k, 3 d_i)$, $\phi(j, 3 d_j) = \phi(k, 3 d_j)$,    
        and $s (2d_k - d_i - d_j)/ d_k < m$.
\end{enumerate} 
\end{theorem}

\begin{proof}
Note that $f(x)$ is \mfqstar \ifa $F(x) := (1/3) f(x)$ is \mfqstar, where
\begin{equation}\label{map_aixri3}
  F(x) = 
  \begin{cases}
  a_0 x^{r_{0}}   & \text{for $x \in C_0$,} \\
  a_1 x^{r_{1}}   & \text{for $x \in C_1$,} \\
  a_2 x^{r_{2}}   & \text{for $x \in C_2$.} \\ 
  \end{cases}   
\end{equation}
By \cref{lem:fi(ci)}, 
$f_t$ is $d_t$-to-$1$ from $C_t$ to $C_{\phi(t, 3)}$,
where $f_t (x) = a_{t}x^{r_t}$ and $t \in \{0, 1, 2\}$.
Let $I_{i j} = f_i(C_i) \cap f_j(C_j)$ for any $i, j \in \{0, 1, 2\}$.
Without loss of generality, we may assume $d_0 \le d_1 \le d_2$. 
We first show the necessity. 
Assume $F$ is \mfqstar. Then 
\[
m \in \{d_0, d_1, d_2, d_0 + d_1, d_0 +d_2, d_1 + d_2, d_0 + d_1 + d_2\}.
\]

Case 1: $m = d_t$ for some $t \in \{0, 1, 2\}$.
If $m = d_t < d_2$ with $t \in \{0, 1\}$, 
then $F$ is not \mfqstar,
since each element in $f_2(C_2)$ has $d_2 > m$ preimages.
Hence $m = d_2$ and thus $I_{02} = I_{12} = \varnothing$.

If $I_{01} = \varnothing$, then $m = d_0 = d_1$. 
Indeed, if $d_t < m$ for some $t \in \{0, 1\}$, 
then $C_t \subseteq E_F(\fqstar)$ and so $\# E_F(\fqstar) \ge s$,
contrary to $\# E_F(\fqstar) < m \le s$.
If $I_{01} \ne \varnothing$, 
then $m = d_0 + d_1$ and $f_1(C_1) \subseteq f_0(C_0)$.
Indeed, if $m > d_0 + d_1$, then $C_0 \cup C_1 \subseteq E_F(\fqstar)$,
contrary to $\# E_F(\fqstar) < s$.
If $m < d_0 + d_1$, then each element in $I_{01}$ has 
$d_0 + d_1 > m$ preimages, contrary to that $F$ is \mfqstar.
If $m = d_0 + d_1$, $I_{01} \ne \varnothing$, 
and $f_1(C_1) \not\subseteq f_0(C_0)$, 
then $F$ is not \mfield{m}{C_0 \cup C_1} by \cref{lem:di+dj},
and thus $F$ is not \mfqstar by $I_{02} = I_{12} = \varnothing$. 
Hence $F(E_F(\fqstar)) =  f_0(C_0) \setminus f_1(C_1)$ and so
\[
  \# E_{F}(\fqstar) 
    = d_0 (s/d_0 - s/d_1)
    = s (d_1 - d_0) / d_1 
    < m.
\]
In summary, by \cref{thm:fiCifjCj}, if $F$ is \mfqstar under the assumption of Case 1, then 
one of the following holds:
\begin{enumerate}[\upshape(1)]
    \item $m = d_2 = d_0 = d_1$ and $\phi(0, 3m)$, $\phi(1, 3m)$, and $\phi(2, 3m)$ are distinct;
    \item $m = d_2 = d_0 + d_1$, $d_0 \mid d_1$, 
        $\phi(0, 3 d_0) = \phi(1, 3 d_0) \ne \phi(2, 3 d_0)$,
         and $s (d_1 - d_0) / d_1 < m$.
\end{enumerate}   

Case 2: $m = d_0 + d_1$.
Then $I_{01} \ne \varnothing$ and $f_2(C_2) \cap I_{01} = \varnothing$.
We can assert that $I_{12} = \varnothing$.
Indeed, if $I_{12} \ne \varnothing$,
then $\# F^{-1}(b) = d_1 + d_2$ for any $b \in I_{12}$.
When $m < d_1 + d_2$, $F$ is not \mfqstar. 
Thus $m = d_1 + d_2$ and so $d_0 = d_1 = d_2$ by $d_0 \le d_1 \le d_2$.
By \cref{f_i(C_i)_di=dj}, 
$I_{01} \ne \varnothing$ implies that $I_{01} = f_0(C_0) = f_1(C_1)$, 
and thus $f_2(C_2) \cap I_{01} = \varnothing$ \ifa $I_{12} = \varnothing$,
contrary to $I_{12} \ne \varnothing$.

We next consider $I_{02}$.
If $I_{02} \ne \varnothing$, 
then $\# F^{-1}(b) = d_0 + d_2$ for any $b \in I_{02}$.
When $m < d_0 + d_2$, $F$ is not \mfqstar.
Thus $m = d_0 + d_2$ and so $d_1 = d_2$.
Let 
\begin{equation}\label{eq:A0A1}
    A_0 = f_0(C_0) \setminus (I_{01} \cup I_{02}) \qtq{and}
    A_1 = (f_1(C_1)\setminus I_{01}) \cup (f_2(C_2)) \setminus I_{02}).
\end{equation}
Then for any $a \in A_0$ and $b \in A_1$, we have
$\# F^{-1}(a) = d_0$ and $\# F^{-1}(b) = d_1$.
Hence $A_0 \cup A_1 = F(E_{F}(\fqstar))$ and so
\begin{equation*}\label{eq:EFd0d1}
    \# E_{F}(\fqstar) = d_0 \# A_0 + d_1 \# A_1 < m.
\end{equation*}
Thus $\# A_0 = 0$ or $\# A_1 = 0$ by $m = d_0 + d_1$.
When $\# A_0 = 0$, we have $I_{01} \cup I_{02} = f_0(C_0)$,
and thus
\begin{equation*}
    2s/\lcm(d_0, d_1) 
   = \# I_{01} + \# I_{02}
   = \# f_0(C_0) = s/d_0 
\end{equation*}
by \cref{thm:fiCifjCj}.
Hence $\lcm(d_0, d_1) = 2d_0$ and so $d_1 = 2 d_0$.
Then $I_{01} \ne \varnothing$ and \cref{thm:fiCifjCj} 
imply that $f_1(C_1) \subset f_0(C_0)$, i.e.,
$I_{01} = f_1(C_1)$, and so $\# I_{01} = s / d_1$.
Hence $\# I_{02} = s / d_1 = s / d_2$ 
and so $I_{02} = f_2(C_2)$.
Thus $\# A_1 = 0$ by \cref{eq:A0A1}. In this case, the size of the exceptional set is $0$.
When $\# A_1 = 0$, we get $I_{0t} = f_t(C_t)$ by \cref{eq:A0A1}, 
i.e., $f_t(C_t) \subseteq f_0(C_0)$, where $t = 1, 2$. Then
\[
    \# E_{F}(\fqstar) = d_0 \# A_0
    = d_0 (s/d_0 - 2s/d_1) = s(d_1 - 2 d_0)/ d_1 < m.
\]

If $I_{02} = \varnothing$,
then $\# F^{-1}(b) = d_2$ for any $b \in f_2(C_2)$.
When $m < d_2$, $F$ is not \mfqstar.
The case $m = d_2$ is already covered in Case~1.
When $m > d_2$, we get $C_2 \subseteq E_{F}(\fqstar)$.
If $f_1(C_1) \not\subseteq f_0(C_0)$, 
then $F$ is not \mfield{m}{C_0 \cup C_1} by \cref{lem:di+dj},
and thus $F$ is not \mfqstar. 
Hence $f_1(C_1) \subseteq f_0(C_0)$ and 
$F(E_F(\fqstar)) = f_2(C_2) \cup (f_0(C_0) \setminus f_1(C_1))$, and so
\begin{equation}\label{eq:EF_d0+d1}
    \# E_{F}(\fqstar) = s + d_0 (s/d_0 - s/d_1)
                      < m.
\end{equation}
Since $m = d_0 + d_1$, we get $d_0 + d_1 > s$, 
and so $d_1 = s$ by $d_1 \mid s$ and $d_0 \le d_1$.
Substituting $d_1 = s$ into \cref{eq:EF_d0+d1} yields 
$d_0 > s/2$, i.e., $d_0 = s$. Then $d_0 = d_1 = d_2 = s$ 
and thus $f_0(C_0) = f_1(C_1)$.

In summary, by \cref{thm:fiCifjCj}, if $F$ is \mfield{(d_0 + d_1)}{\fqstar},
then one of the following holds:
\begin{enumerate}[resume*,start=2]
    \item $d_0 + d_1 = d_2$, $d_0 \mid d_1$, 
          $\phi(0, 3 d_0) = \phi(1, 3 d_0) \ne \phi(2, 3 d_0)$,
          and $s (d_1 - d_0) / d_1 < m$;
    \item $d_0 \mid d_1$, $d_1 = d_2$, 
          $\phi(0, 3 d_0) = \phi(1, 3 d_0) = \phi(2, 3 d_0)$,
          $\phi(1, 3 d_1) \ne \phi(2, 3 d_1)$,
          and $s(d_1 - 2 d_0)/ d_1 < m$;
    \item $d_0 = d_1 = d_2 = s$ and $\phi(0, 3 s) = \phi(1, 3 s) \ne \phi(2, 3 s)$.
\end{enumerate} 

Case 3: $m = d_0 + d_2$. 
Then $I_{02} \ne \varnothing$ and $f_1(C_1) \cap I_{02} = \varnothing$.
We can assert that $I_{12} = \varnothing$. 
Indeed, if $I_{12} \ne \varnothing$,
then $\# F^{-1}(b) = d_1 + d_2$ for any $b \in I_{12}$.
When $m < d_1 + d_2$, $F$ is not \mfqstar. 
Thus $m = d_1 + d_2$ and so $d_0 = d_1$.
By \cref{f_i(C_i)_di=dj}, 
$f_1(C_1) \cap I_{02} = \varnothing$ implies that $I_{01} = \varnothing$. 
Let 
\begin{equation}\label{eq:B0B1}
    B_0 = f_2(C_2) \setminus (I_{02} \cup I_{12}) \qtq{and}
    B_1 = (f_0(C_0)\setminus I_{02}) \cup (f_1(C_1)) \setminus I_{12}).
\end{equation}
Then for any $a \in B_0$ and $b \in B_1$, we have
$\# F^{-1}(a) = d_2$ and $\# F^{-1}(b) = d_0$.
Hence $F(E_{F}(\fqstar)) = B_0 \cup B_1$, and so
\begin{equation*}
    \# E_{F}(\fqstar) = d_2 \# B_0 + d_0 \# B_1 < m.
\end{equation*}
Thus $\# B_0 = 0$ or $\# B_1 = 0$ by $m = d_0 + d_2$.
When $\# B_0 = 0$, we get $I_{02} \cup I_{12} = f_2(C_2)$ by \cref{eq:B0B1},
and thus  
\[
   2s/\lcm(d_0, d_2)
   = \# I_{02} + \# I_{12}
   = \#f_2(C_2) = s/d_2
\]
by \cref{thm:fiCifjCj}.
Hence $\lcm(d_0, d_2) = 2 d_2$ and so $d_0 = 2 d_2$, 
contrary to $d_0 \le d_2$.
When $\# B_1 = 0$, we have 
$I_{t2} = f_t(C_t)$ by \cref{eq:B0B1}, i.e.,
$f_t(C_t) \subseteq f_2(C_2)$,
where $t = 0, 1$.
Since $\# f_t(C_t) = s/d_t \ge s/d_2 = \# f_2(C_2)$,
we obtain $f_0(C_0) = f_1(C_1) = f_2(C_2)$,
contrary to $f_1(C_1) \cap I_{02} = \varnothing$.

We next consider $I_{01}$.
If $I_{01} \ne \varnothing$, 
then $\# F^{-1}(b) = d_0 + d_1$ for any $b \in I_{01}$.
The case $m = d_0 + d_1$ is already covered in Case~2.
Thus we need only consider the case $m > d_0 + d_1$.
Let $B_2 = f_0(C_0) \setminus I_{02}$
and $B_3 =  f_2(C_2) \setminus I_{02}$.
Then $f_1(C_1) \cup B_2 \cup B_3 = F(E_F(\fqstar))$
and so 
\[
   \# E_{F}(\fqstar) = s + d_0 \# B_2 + d_2 \# B_3 < m. 
\]
Thus $\# B_2 = 0$ or $\# B_3 = 0$ by $m = d_0 + d_2$,
i.e., $I_{02} = f_0(C_0)$ or $f_2(C_2)$.
Since $m > d_0 + d_1$ and $d_1 \ge d_0$, we have $d_2 > d_0$, and so 
$\# f_2(C_2) = s/d_2 < s/d_0 = \# f_0(C_0)$.
Hence $I_{02} = f_2(C_2)$ and thus 
\begin{equation}\label{eq:EF_d0+d2}
     \# E_{F}(\fqstar) = s + d_0 \# B_2
                       = s + d_0 (s/d_0 - s/d_2)
                       < m.
\end{equation}
Since $m = d_0 + d_2$, we get $d_0 + d_2 > s$,
and so $d_2 = s$ by $d_2 \mid s$ and $d_0 \le d_2$.
Substituting $d_2 = s$ into \cref{eq:EF_d0+d2},
we have $d_0 > s/2$, i.e., $d_0 = s$,
contrary to $d_2 > d_0$.

If $I_{01} = \varnothing$, 
then $C_1 \subseteq E_{F}(\fqstar)$ by $d_1 < m$.
When $f_2(C_2) \not\subseteq f_0(C_0)$, 
$F$ is not \mfield{m}{C_0 \cup C_2} by \cref{lem:di+dj},
and thus $F$ is not \mfqstar. 
Then $f_2(C_2) \subseteq f_0(C_0)$, i.e., $I_{02} = f_2(C_2)$,
and so $B_3 = 0$.
Hence $F(E_F(\fqstar)) = f_1(C_1) \cup B_2$ and so
\cref{eq:EF_d0+d2} holds.
Then $d_0 = d_1 = d_2 = s$ and thus $f_0(C_0) = f_2(C_2)$. 

In summary, by \cref{thm:fiCifjCj}, if $F$ is \mfield{(d_0 + d_2)}{\fqstar}, then
$d_0 = d_1 = d_2 = s$ and $\phi(0, 3 s) = \phi(2, 3 s) \ne \phi(1, 3 s)$.

Case 4: $m = d_1 + d_2$. 
Then $I_{12} \ne \varnothing$ and $f_0(C_0) \cap I_{12} = \varnothing$.
We can assert that $I_{01} = I_{02} = \varnothing$. 
Indeed, if $I_{0t} \ne \varnothing$ for some $t \in \{1, 2\}$,
then $\# F^{-1}(b) = d_0 + d_t$ for any $b \in I_{0t}$.
 If $m = d_0 + d_t$, there exists a contradiction, 
which has been shown in Cases~2 and~3.
Thus we need only consider the case $m > d_0 + d_t$.
Let $D_1 = f_1(C_1) \setminus I_{12}$ and $D_2 = f_2(C_2) \setminus I_{12}$.
Then $f_0(C_0) \cup D_1 \cup D_2 = F(E_F(\fqstar))$,
and thus 
\[
   \# E_{F}(\fqstar) = s + d_1 \# D_1 + d_2 \# D_2 < m. 
\]
Thus $\# D_1 = 0$ or $\# D_2 = 0$ by $m = d_1 + d_2$. 
When $\# D_1 = 0$, we get $I_{12} = f_1(C_1)$, 
and so $I_{12} = f_1(C_1) = f_2(C_2)$ by $\# f_1(C_1) \ge \# f_2(C_2)$.
Hence $f_0(C_0) \cap I_{12} = \varnothing$ implies that $I_{0t} = \varnothing$,
contrary to $I_{0t} \ne \varnothing$.
When $\# D_2 = 0$, we have $I_{12} = f_2(C_2)$, and thus
\begin{equation}\label{eq:EF_d1+d2}
   \# E_{F}(\fqstar) = s + d_1 \# D_1
                     = s + d_1 (s/d_1 - s/d_2)
                     < m. 
\end{equation}
Since $m = d_1 + d_2$, we obtain $d_1 + d_2 > s$, 
and so $d_2 = s$ by $d_2 \mid s$ and $d_1 \le d_2$.
Substituting $d_2 = s$ into \cref{eq:EF_d1+d2} yields $d_1 > s/2$, i.e., $d_1 = s$.
Thus $I_{12} = f_1(C_1) = f_2(C_2)$,
and so $f_0(C_0) \cap I_{12} = \varnothing$ implies that $I_{0t} = \varnothing$,
contrary to $I_{0t} \ne \varnothing$.

Since $I_{01} = I_{02} = \varnothing$, 
we have $C_0 \subseteq E_{F}(\fqstar)$ by $d_0 < m$.
If $f_2(C_2) \not\subseteq f_1(C_1)$, 
then $F$ is not \mfield{m}{C_1 \cup C_2} by \cref{lem:di+dj},
and thus $F$ is not \mfqstar. 
Then $f_2(C_2) \subseteq f_1(C_1)$, i.e., $I_{12} = f_2(C_2)$,
and so $\# D_2 = 0$.
Hence $F(E_F(\fqstar)) = f_0(C_0) \cup D_1$ and so \cref{eq:EF_d1+d2} holds. 
Then $d_1 = d_2 = s$ and thus $f_1(C_1) = f_2(C_2)$.

In summary, by \cref{thm:fiCifjCj}, if $F$ is \mfield{(d_1 + d_2)}{\fqstar},
then $d_1 = d_2 = s$, $\phi(1, 3 s) = \phi(2, 3 s)$, and
$\phi(0, 3 d_0) \ne \phi(1, 3 d_0)$.

Case 5: $m = d_0 + d_1 + d_2$.
Then 
\begin{equation}\label{eq_f0capf1capf2}
     I \coloneq f_0(C_0) \cap f_1(C_1) \cap f_2(C_2) \ne \varnothing.
\end{equation}
Denote $A_i = f_i (C_i) \setminus I$ for any $i \in \{0, 1, 2\}$.
If $A_i \ne \varnothing$ for any $i \in \{0, 1, 2\}$,
then $\# f_i^{-1}(b_i) = d_i$ for any $b_i \in A_i$. 
Hence $A_0 \cup A_1 \cup A_2 = F(E_{F}(\fqstar))$ by \cref{eq_f0capf1capf2},
and so $\# E_{F}(\fqstar) \ge d_0 + d_1 + d_2 = m$, 
contrary to the definition of \mtoone.
Hence there exists $i \in \{0, 1, 2\}$ such that $A_i = \varnothing$, 
i.e., $I = f_i(C_i)$. 
Moreover, since $d_0 \le d_1 \le d_2$, 
we have $\# f_0 (C_0) \ge \# f_1 (C_1) \ge \# f_2 (C_2)$ and so
$I = f_2 (C_2)$.
Then \cref{eq_f0capf1capf2} implies that 
$f_2(C_2) \subseteq f_0(C_0) \cap f_1(C_1)$; that is,
$f_2(C_2) \subseteq f_0(C_0)$ and $f_2(C_2) \subseteq f_1(C_1)$. 
By \cref{thm:fiCifjCj}, it is equivalent to
$d_0 \mid d_2$, $d_1 \mid d_2$,
$\phi(0, 3 d_0) = \phi(2, 3 d_0)$, and $\phi(1, 3 d_1) = \phi(2, 3 d_1)$.
Moreover, we get $F(E_{F}(\fqstar)) = A_0 \cup A_1$ and thus
\[
 \# E_{F}(\fqstar) 
 = d_0 (s/d_0 - s/d_2) + d_1 (s/d_1 - s/d_2)
 = s (2 d_2 - d_0 - d_1)/ d_2 < m.  
\]

According to the analysis above,
it is easy to verify that the sufficiency holds.
\end{proof}

\cref{aixrimto1_3pieces} completely characterizes the many-to-one 
property of the generalized cyclotomic mappings for $\ell = 3$.
In \cref{item:L=3-1,item:L=3-2,item:L=3-3,item:L=3-6}, 
$m$ can be very small for any $q$. 
However, $m = 2 s$ in \cref{item:L=3-4,item:L=3-5}, 
and thus $m$ is directly proportional to $q$ in these cases. 

When $m = 1$ and $a_0^s = a_1^s = a_2^s$,
\cref{aixrimto1_3pieces} is reduced to
\cite[Proposition~2.6]{Wang-cyc}.
Applying \cref{aixrimto1_3pieces} to $m = 2$ or~$3$ 
yields the following results.

\begin{corollary} \label{aixri2to1_3pieces}
Let the notation and the hypotheses 
be as in \cref{aixrimto1_3pieces}.
If $q \ge 7$, then $f$ is \mfield{2}{\fqstar} \ifa 
one of the following holds:
\begin{enumerate}[\upshape(1)]
  \item $d_0 = d_1 = d_2 = 2$ and 
        $\phi(0, 6)$, $\phi(1, 6)$, and $\phi(2, 6)$ are distinct; 
  \item $d_i = d_j = 1$, $d_k = 2$,
        and $\phi(i, 3) = \phi(j, 3) \ne \phi(k, 3)$. 
\end{enumerate} 
\end{corollary}

\begin{corollary} \label{aixri3to1_3pieces}
Let the notation and the hypotheses 
be as in \cref{aixrimto1_3pieces}.
If $q \ge 19$, 
then $f$ is \mfield{3}{\fqstar} 
\ifa one of the following holds:
\begin{enumerate}[\upshape(1)]
    \item $d_0 = d_1 = d_2 = 3$ and $\phi(0, 9)$, $\phi(1, 9)$, 
            and $\phi(2, 9)$ are distinct; 
    \item $d_i = 1$, $d_j = d_k = 2$,
          $\phi(0, 3) = \phi(1, 3) = \phi(2, 3)$,
          and $\phi(j, 6) \ne \phi(k, 6)$.
    \item $d_0 = d_1 = d_2 = 1$ and
          $\phi(0, 3) = \phi(1, 3) = \phi(2, 3)$.
\end{enumerate}
\end{corollary}

\cref{aixri3to1_3pieces} is a generalization of 
\cite[Theorem~4.16]{nto1-NiuLQL23} where 
 $r_0 = r_1 = r_2$ and $(r_0, s) = 1$.
 
In \cref{aixri2to1_3pieces}, 
take $q = 13$, $s = 4$, and $\xi = 2$. Then $\omega = 3$.
Take $(a_0, a_1, a_2) = (1, 2, -5)$ and all $r_i = 2$.
Then all $d_i = 2$, $\phi(0, 6) = 0$, 
$\phi(1, 6) = 3$, and $\phi(2, 6) = 1$.
Take $(a_0, a_1, a_2) = (1, 4, 3)$ 
and $(r_0, r_1, r_2) = (1, 1, 2)$.
Then $d_0 = d_1 = 1$, $d_2 = 2$, 
$\phi(0, 3) = \phi(1, 3) = 0$, and $\phi(2, 3) = 2$.
Hence we obtain the next example.

\begin{example}\label{ex:2to1F13}
The polynomials
$x^{10} + 4 x^{6} - 2 x^2$ and 
$x^{10} - 4 x^{6} - 2 x^{5} + 3 x^{2} + 5 x$
are \mfield{2}{\f_{13}^{*}}.
\end{example}




In \cref{aixri3to1_3pieces}, 
take $q = 64$ and $s = 21$. Then $\omega = \xi^{21}$.
Take $(a_0, a_1, a_2) = (\xi, \xi^{14}, \xi^{35})$ and all $r_i = 3$.
Then all $d_i = 3$, $\phi(0, 9) = 1$, 
$\phi(1, 9) = 8$, and $\phi(2, 9) = 5$.
Take $(a_0, a_1, a_2) = (\xi^{12}, \xi^{2}, \xi^{25})$ 
and all $r_i = 1$.
Then all $d_i = 1$ and $\phi(i, 3) = 0$.
Hence we obtain the next example.

\begin{example}
The polynomials 
$x^{45} + \xi x^{24} + x^3$ and $x^{43} + \xi^3 x^{22} + \xi^5 x$
are \mfield{3}{\f_{64}^{*}}, where $\xi$ is a primitive element of $\f_{64}$.
\end{example}

\section{Multiple branches}
\label{sec_ell>=3}

In this section, we consider the case~$f$ in 
\cref{map_aixri} has $\ell$ branches on $\fqstar$.
We first determine the \twotoone property of~$f$.

\begin{theorem}\label{Lpieces2to1}
Let $q-1 = \ell s$ with $2 \le \ell \le (q-1)/2$ and
\begin{equation*}
  f(x) = \frac{1}{\ell}
    \sum_{i=0}^{\ell-1}\sum_{j=0}^{\ell-1}
    a_{i}  \omega^{- i j} x^{r_i + j s},
\end{equation*}
where $a_i \in \fqstar$, $r_i\in \n$, and $\omega$ is a primitive 
$\ell$-th root of unity over $\fq$. 
Let $I$ and $T$ be subsets of 
$[\ell] \coloneq \{0, 1, \ldots, \ell-1\}$ 
such that $(r_i, s) = 1$ for any $i \in I$ 
and $(r_t, s) = 2$ for any $t \in T$. 
Then $f$ is \mfield{2}{\fqstar} \ifa one of the following holds:
\begin{enumerate}[\upshape(1)]
\item $I = [\ell]$, $\# I$ is even, and $\phi(x, \ell)$ is \mset{2}{I};
\item $T = [\ell]$ and $\phi(x, 2 \ell)$ is \mset{1}{T};
\item \label{L2to1(3)} 
        $I \ne [\ell]$, $T \ne [\ell]$, $I \cup T = [\ell]$,
        $\phi(i, \ell) \ne \phi(t, \ell)$ 
        for any $i \in I$ and $t \in T$, 
        $\# I$ is even and $\phi(x, \ell)$ is \mset{2}{I},
        and $\phi(x, 2 \ell)$ is \mset{1}{T}.
\end{enumerate}
\end{theorem}

\begin{proof}
By \cref{map_aixri}, 
$f(x) = f_i(x) \coloneq a_i x^{r_i}$ for $x \in C_i$.
We first show the necessity.
Assume $f$ is \mfield{2}{\fqstar}.

Case~1: $I = [\ell]$, i.e., 
$(r_i, s) = 1$ for $0 \le i \le \ell-1$.
Then $f_i$ is \onetoone from $C_i$ to 
$C_{\phi(i, \ell)}$ for $i \in I$,
where $\# C_i = s \ge 2$ by $\ell \le (q-1)/2$.
Thus $f_b(C_b) = f_c(C_c)$ or 
$f_b(C_b) \cap f_c(C_c)= \varnothing$
for any $b, c \in I$.
If there exists $b \in I$ such that 
$f_c(C_c) \cap f_b(C_b) = \varnothing$ 
for any $c \in I \setminus \{ b\}$,
then $C_b \subseteq E_{f}(\fqstar)$, 
and so $\# E_{f}(\fqstar) \ge 2$,
contrary to that $f$ is \mfield{2}{\fqstar}.
Thus for any $b \in I$,
there is a unique $c \in I \setminus \{ b\}$ 
such that $f_b(C_b) = f_c(C_c)$ or equivalently
$\phi(b, \ell) = \phi(c, \ell)$ by \cref{thm:fiCifjCj}.
That is, $\# I$ is even and $\phi(x, \ell)$ is \mset{2}{I}.

Case~2: $T = [\ell]$, i.e., 
$(r_t, s) = 2$ for $0 \le t \le \ell-1$. 
Then $f_t$ is \mset{2}{C_t} for any $t \in T$.
Since $f$ is \mfield{2}{\fqstar}, we obtain
$f_u(C_u) \cap f_v(C_v) = \varnothing$, or equivalently
$\phi(u, 2 \ell) \ne \phi(v, 2 \ell)$ by \cref{thm:fiCifjCj}, 
for any distinct $u, v \in T$.
Thus $\phi(x, 2 \ell)$ is \mset{1}{T}.

Case~3: $I \ne [\ell]$ and $T \ne [\ell]$.
If $(r_i, s) \ge 3$ for some $i \in [\ell]$, 
then $f_i$ is $(r_i, s)$-to-$1$ 
from $C_i$ to $C_{\phi(i, \ell)}$,
and so $f$ is not \mfield{2}{\fqstar}, a contradiction. 
Thus $(r_i, s) \le 2$ and so $I \cup T = [\ell]$.
Since $f$ is \mfield{2}{\fqstar}, we have
$f_i(C_i) \cap f_t(C_t) = \varnothing$, or equivalently 
$\phi(i, \ell) \ne \phi(j, \ell)$ by \cref{thm:fiCifjCj},
for any $i \in I$ and $t \in T$.
The remainder of the argument is the same as that in Cases~1 and~2,
and so is omitted.

According to the analysis above,
it is easy to verify that the sufficiency holds.
\end{proof}

\cref{Lpieces2to1} gives a criterion for~$f$ to be 
\mfield{2}{\fqstar} when $2 \le \ell \le (q-1)/2$.
If $\ell = 1$, then $f(x) = a_0 x^{r_0}$ and so 
$f$ is \mfield{2}{\fqstar} \ifa $(r_0, q-1) = 2$.
If $\ell = q - 1$, then $I = [\ell]$ and thus~$f$ 
is \mfield{2}{\fqstar} \ifa $\phi(x, \ell)$ is \mset{2}{[\ell]}.
Hence the \twotoone property of~$f$ on $\fqstar$ is completely 
characterized for any divisor $\ell$ of $q-1$.

In \cref{Lpieces2to1}, take $q = 17$, $\ell = s = 4$, 
and $\xi = 3$. Then $1/\ell = \omega = - 4$.
Take 
\[
(a_0, a_1, a_2, a_3) = (-6, 4, -3, -2)  \qtq{and}
(r_0, r_1, r_2, r_3) = (1, 1, 1, 1).
\]
Then $I = [4]$ and 
$\phi(x, \ell) = (x + \log_{\xi} a_x) \bmod{4}$.
Thus $\phi(x, \ell)$ is \mfield{2}{I}. Take 
\[
(a_0, a_1, a_2, a_3) = (-6, -8, -3, 1)  \qtq{and}
(r_0, r_1, r_2, r_3) = (2, 2, 2, 2).
\]
Then $T = [4]$ and 
$\phi(x, 2\ell) = (2 x + \log_{\xi} a_x) \bmod{8}$. 
Thus $\phi(x, 2\ell)$ is \mfield{1}{T}. Take 
\[
(a_0, a_1, a_2, a_3) = (8, 2, -8, 4)  \qtq{and}
(r_0, r_1, r_2, r_3) = (2, 3, 2, 3).
\]
Then $I = \{1, 3\}$ and $T = \{0, 2\}$.
It is easy to verify that all the conditions 
in \cref{L2to1(3)} of \cref{Lpieces2to1} hold. 
Hence we get the next example.

\begin{example}
The following polynomials are \mfield{2}{\f_{17}^{*}}:
\begin{align*}
   F_1(x) & = 6 x^{13} - 7 x^9 + x^5 - 6 x, \\
   F_2(x) & = 4 x^{14} + 8 x^{10} + 3 x^6 - 4 x^2, \\
   F_3(x) & = 2 x^{15} + 4 x^{14} + 7 x^{11} - 2 x^7 + 4 x^6 - 7 x^3.
\end{align*}
\end{example}

We now consider the \mtoone property of~$f$.  
If $f_i(C_i) \cap f_j(C_j) = \varnothing$ 
for any distinct $i, j \in [\ell]$, 
then $f$ is \mfield{m}{\fqstar} \ifa 
$(r_i, s) = m$ for any $i \in [\ell]$. 
If $(r_i, s) = 1$ for any $i \in [\ell]$ and 
$m \mid \ell$, then $f$ is \mfield{m}{\fqstar} 
\ifa $\phi(x, \ell)$ is \mset{m}{[\ell]}. 
If $\ell \ge 4$, $f_{i_0}(C_{i_0}) = f_{j_0}(C_{j_0})$ 
for two fixed $i_0$, $j_0 \in [\ell]$, and 
$f_i(C_i) \cap f_j(C_j) = \varnothing$ 
for any $i, j \in [\ell] \setminus \{i_0\}$, 
then $f$ is \mfield{m}{\fqstar} \ifa 
$m = 2(r_{i_0}, s)$ and $m = (r_i, s)$ for any 
$i \in [\ell] \setminus \{i_0, j_0\}$.

In general, it is difficult to determine the \mtoone property
of~$f$ in \cref{map_aixri} for $\ell \ge 4$ and $m \ge 3$.
However, this problem becomes relatively simple when all $(r_i, s)$'s are equal.

\begin{theorem}\label{Lpiecesmto1}
Let $q-1 = \ell s$ and $[\ell] = \{0, 1, \ldots, \ell-1\}$.
Let
\begin{equation*}
  f(x) = \frac{1}{\ell}
    \sum_{i=0}^{\ell-1}\sum_{j=0}^{\ell-1}
    a_{i}  \omega^{- i j} x^{r_i + j s},
\end{equation*}
where $a_i \in \fqstar$, $r_i\in \n$, and $\omega$ is a primitive 
$\ell$-th root of unity over $\fq$. 
Assume $(r_i, s) = d$ for any $i \in [\ell]$.
Then for $1 \le m < q$, $f$ is \mfield{m}{\fqstar} \ifa 
$d \mid m$, 
$\phi(x, \ell d)$ is \mset{(m/d)}{[\ell]},
 and $s(\ell \bmod{(m/d)}) < m$.
\end{theorem}

\begin{proof}
By \cref{map_aixri}, 
$f(x) = f_i(x) \coloneq a_i x^{r_i}$ for $x \in C_i$.
For any $i \in [\ell]$, since $(r_i, s) = d$,
each $f_i$ is \mset{d}{C_i} and $E_{f_i}(C_i) = \varnothing$,
and so
each element in $f_i(C_i)$ has $d$ preimages in $C_i$.
By \cref{f_i(C_i)_di=dj},  
$f_i(C_i) = f_j(C_j)$ or 
$f_i(C_i) \cap f_j(C_j) = \varnothing$
for any $i, j \in [\ell]$.
Hence the number of preimages of any $b \in f(\fqstar)$ 
is a multiple of $d$. 
We first show the necessity.
Assume $f$ is \mfield{m}{\fqstar} 
and $q - 1 = u m + v$ with $0 \le v < m$.
Then there are $u$ distinct elements $b_1, b_2, \ldots, b_u \in f(\fqstar)$ 
such that each $b_i$ has exactly~$m$ preimages in~$\fqstar$.
Thus $d \mid m$. 
Let $m_1 = m/d$ and $\ell = k m_1 + t$ with $0 \le t < m_1$. 
Since $b_1$ has $0$ or $d$ preimages in each $C_i$,
then there are exactly $m_1$ sets 
$C_{11}, \ldots, C_{1 m_1}$ 
such that each set has exactly $d$ preimage of $b_1$, 
that is,
\[
   b_1 \in f_{11}(C_{11}) \cap\cdots\cap f_{1 m_1}(C_{1 m_1}). 
\]
Recall that $f_{1i}(C_{1i})$ and $f_{1j}(C_{1j})$ are either equal or disjoint.
Thus $f_{11}(C_{11})  = \cdots = f_{1 m_1}(C_{1 m_1})$.
Similarly, if some $b_i \notin f_{11}(C_{11})$,
then there are another $m_1$ sets $C_{21}, \ldots, C_{2 m_1}$ such that
$f_{21}(C_{21})  = \cdots = f_{2 m_1}(C_{2 m_1})$.
Hence the preimages of $b_1$, $b_2$,~$\ldots$, $b_u$ 
contain at most $k m_1$ sets of the form $C_{ij}$
and so $um \le k m_1 s$, i.e., $u \le ks/d$.
Note that 
\begin{equation}\label{eq_q-1_factor}
     u m + v = q - 1 = \ell s = (k m_1 + t)s = (ks/d) m + st.
\end{equation}
If $st \ge m$, then $u > ks/d$, contrary to $u \le ks/d$.
Thus $st < m$, i.e., $s(\ell \bmod{(m/d)}) < m$.
By \cref{eq_q-1_factor}, 
$v = st$ and $um = (ks/d) m = k m_1 s$.
Hence the preimages of $b_1$, $b_2$,~$\ldots$, $b_u$ 
contain $k m_1$ sets and the remaining~$t$ sets consist of 
$(q - 1) - um = v$ elements.
Thus the image  $f(\fqstar)$ is divided into some disjoint sets: 
\begin{equation*}
    f_{01}(C_{01}) \cup\cdots\cup f_{0t}(C_{0t}) \qtq{and}
    f_{i1}(C_{i1})  = \cdots = f_{i m_1}(C_{i m_1}),  
    \quad\text{$1 \le i \le k$.} 
\end{equation*}           
Therefore, by \cref{f_i(C_i)_di=dj}, we get 
\begin{equation*}
    \{\phi(01, \ell d), \ldots,  \phi(0t, \ell d)\}, \quad 
    \{\phi(i1, \ell d) = \cdots = \phi(i m_1, \ell d) \}, \quad
    \text{$1 \le i \le k$,} 
\end{equation*}
and they are disjoint.
Hence $\phi(x, \ell d)$ is \mset{m_1}{[\ell]}. 

According to the analysis above,
it is easy to verify that the sufficiency holds.
\end{proof}

In \cref{Lpiecesmto1}, take $q = 17$, $\ell = s = 4$, 
and $\xi = 3$. Then $1/\ell = \omega = - 4$.
Take 
\[
(a_0, a_1, a_2, a_3) = (5, 7, -1, -2)  \qtq{and}
(r_0, r_1, r_2, r_3) = (2, 2, 2, 2).
\]
Then $(r_i, s) = 2$ for any $i \in [4]$ and 
$\phi(x, 2 \ell) = (2 x + \log_{\xi} a_x) \bmod{8}$. 
Thus $\phi(x, 2\ell)$ is \mfield{2}{[4]}. Take 
\[
(a_0, a_1, a_2, a_3) = (-2, -6, -4, 3)  \qtq{and}
(r_0, r_1, r_2, r_3) = (3, 3, 1, 3).
\]
Then $(r_i, s) = 1$ for any $i \in [4]$.
Note that $\phi(x, \ell)$ is \mfield{4}{[4]}. 
Hence we get the next example.
\begin{example}
The following polynomials are \mfield{4}{\f_{17}^{*}}:
\begin{align*}
   F_1(x) & = x^{14} + 4 x^{10} + 2 x^6 - 2 x^2, \\
   F_2(x) & = x^{13} - 4 x^{11} - x^9 - x^7 + x^5 + 3 x^3 - x.
\end{align*}
\end{example}

Applying \cref{Lpiecesmto1} to $a_1 = \cdots = a_{\ell -1}$ and 
$r_1 = \cdots = r_{\ell -1}$ yields the following result.

\begin{corollary}\label{Lpiecesmto1_d=1_r1=rl-1}
Let $q-1 = \ell s$ with $\ell \ge 3$.
Let $r_0, r_1 \in \n$ and $a_0, a_1 \in \fqstar$ satisfy $(r_0,s) = (r_1,s) = d$
and $a_0^{s/d} = a_1^{s/d}$.
Then for $1 \le m < q$,
\[
  f(x) = (1/\ell)(a_0 x^{r_0} - a_1 x^{r_1})
         (1 +x^{s} + \cdots + x^{(\ell-1)s}) + a_1 x^{r_1}
\]
is \mfield{m}{\fqstar} \ifa 
$(r_1, \ell d) = m$.
\end{corollary}

\begin{proof}
Let $[\ell] = \{0, 1, \ldots, \ell-1\}$ and $[\ell]^{*} = [\ell] \setminus\{0\}$.
Note that 
\begin{equation*}
  \phi(x, \ell d) = 
  \begin{cases}
    \log_{\xi} a_0 \bmod{\ell d} & \text{for $x = 0$,} \\
    (r_1 x + \log_{\xi} a_1) \bmod{\ell d} & \text{for $x \in [\ell]^{*}$.}
  \end{cases}   
\end{equation*}
The condition $a_0^{s/d} = a_1^{s/d}$ implies that 
$\log_{\xi} {a_0} \bmod{\ell d} = \log_{\xi} {a_1} \bmod{\ell d}$, 
and so
$\phi(x, \ell d) = (r_1 x + \log_{\xi} a_1) \bmod{\ell d}$ for any $x \in [\ell]$.
Hence $\phi(x, \ell d)$ is \mset{(m/d)}{[\ell]} \ifa 
$r_1 x  \bmod{\ell d}$ is \mset{(m/d)}{[\ell]},
or equivalently $(r_1 / d) x  \bmod{\ell}$ is \mset{(m/d)}{[\ell]},
since $r_1 x_1 \equiv r_1 x_2 \pmod{\ell d}$ \ifa 
$(r_1 / d) x_1 \equiv (r_1 / d) x_2 \pmod{\ell}$
for any $x_1, x_2 \in [\ell]$.
It is also equivalent to $(r_1 / d) x$ is \mset{(m/d)}{\z_{\ell}}; 
that is $(r_1/d, \ell) = m/d$, i.e., $(r_1, \ell d) = m$.
Then the result follows from \cref{Lpiecesmto1}.
\end{proof}

Applying \cref{Lpiecesmto1} to $m = d$ or $d = 1$ 
yields the following results.

\begin{corollary}\label{Lpiecesmto1_dto1}
Let the notation and the hypotheses be as in \cref{Lpiecesmto1}.
Then $f$ is \mfield{d}{\fqstar} \ifa 
$\phi(x, \ell d)$ is \mset{1}{[\ell]}.
\end{corollary}

\begin{corollary}\label{Lpiecesmto1_d=1}
With the notation and the hypotheses of \cref{Lpiecesmto1},
let $d = 1$.
Then for $1 \le m < q$, $f$ is \mfield{m}{\fqstar} \ifa 
$\phi(x, \ell)$ is \mset{m}{[\ell]} and $s(\ell \bmod{m}) < m$.
\end{corollary}

\cref{Lpiecesmto1_dto1} generalizes \cite[Corollary~2.1]{Qin2469} 
where $\phi(x, \ell)$ is \mset{1}{[\ell]} and only the sufficient condition is given.
\cref{Lpiecesmto1_d=1} generalizes \cite[Corollary~2.2]{Qin2469} 
where $m \mid \ell$ and only the sufficient condition is given.
Applying \cref{Lpiecesmto1_d=1} to $a_i = \ell \omega^{i(\ell - 1)}$ and 
$r_i = q^i$ yields the following result.

\begin{corollary}\label{Lpiecesmto1_d=1_ex}
Let $q^n \equiv 1 \pmod{\ell^2}$ and $s = (q^n -1)/\ell$.
Let 
\[
    f(x) = \sum_{i = 0}^{ \ell -1} x^{q^i}
           \prod_{j=0, j \ne i}^{\ell -1} (x^s - \omega^j),
\]
where $\omega$ is a primitive 
$\ell$-th root of unity over $\fqn$.
Then for $1 \le m < q^n$, 
$f$ is \mfield{m}{\fqnstar} \ifa
$(x q^x) \bmod{\ell}$ is \mset{m}{[\ell]} and $s(\ell \bmod{m}) < m$.
\end{corollary}

\begin{proof}
Clearly, $f(x) = a_i x^{q^i}$ for $x \in C_i$, where 
\[
    a_i = \prod_{j=0, j \ne i}^{\ell -1} (\omega^i - \omega^j)
        = \omega^{i(\ell - 1)} (1- \omega^{-1}) \cdots (1- \omega^{-(\ell -1)})
        = \ell \omega^{i(\ell - 1)}.
\]
Since $(q^i, s) = 1$ for any $i \in [\ell]$,
it follows that $f$ is \mfield{m}{\fqnstar} \ifa 
$\phi(x, \ell)$ is \mset{m}{[\ell]} and $s(\ell \bmod{m}) < m$
by \cref{Lpiecesmto1_d=1}. 
Because $q^n \equiv 1 \pmod{\ell^2}$, 
we have $ \ell \mid s$, and so 
$\log_{\xi} \omega^{i(\ell - 1)} = i(\ell - 1)s \equiv 0 \pmod{\ell}$. 
Hence $\phi(x, \ell)$ is \mset{m}{[\ell]} \ifa 
$(x q^x) \bmod{\ell}$ is \mset{m}{[\ell]}.
\end{proof}

\cref{Lpiecesmto1_d=1_ex} generalizes
\cite[Propositions~2.9]{Wang-cyc} where $m = 1$.

\section{Realization of constants by polynomials}
\label{sec_Realization}

In this section,
we will use polynomials of specific forms to represent the constants $a_{0}$, $a_{1}$, \ldots, $a_{\ell-1}$ in the generalized cyclotomic mappings.
That is, each $a_{i}$ can be realized by the polynomial of the form $h_i(x^s)$
since $h_i(x^s) = h_i(\omega^i)$ if $x \in C_i$,
where $\omega$ is a primitive 
$\ell$-th root of unity over $\fq$. 
Then substituting $a_i$ by $h_i(x^s)$ 
in \cref{poly_aixri} yields 
\begin{equation}\label{eq:f=xrihixs}
  f(x) = \frac{1}{\ell}
    \sum_{i=0}^{\ell-1} x^{r_i} h_i(x^s)
    \sum_{j=0}^{\ell-1} \omega^{- i j} x^{j s}.
\end{equation}
We next present some \mtoone mappings of this form
by the results in the previous sections. 

\subsection{ Generalized cyclotomic mappings}

In general, the discrete logarithm $\log_{\xi} a_i$ 
in $\phi(i, n)$ is hard to compute when $q$ is very large. 
We first employ \cref{aixrimto1_2pieces} to characterize 
two classes of \mtoone mappings without involving discrete logarithms.
Similar results on the \onetoone mappings can be found 
in \cite{Wang-cyc}.

\begin{corollary}\label{aixrimto1_gi(x)}
Let $q$ be odd, $s = (q-1)/2$, and $g_0, g_1 \in \fqx$ with $g_0(1) g_1(-1) \ne 0$.
Let 
\[
  f(x) = x^{r_0} g_0(x^s)^{2 d_0} (1 + x^s) + x^{r_1} g_1(x^s)^{2 d_1} (1 - x^s),
\]
where $r_i \in \n$ and $d_i = (r_i, s)$ with $i \in \{0, 1\}$.
Then for $1 \le m < q$, 
$f$ is \mfqstar \ifa one of the following holds:
\begin{enumerate}[\upshape(1)]
  \item $m = d_0 = d_1$ and $r_1$ is odd;
  \item $m = d_0 + d_1$, $d \mid m$,
        $r_1$ is even,
        and $s (m - 2d)/(m - d) < m$, 
        where $d = \min \{d_0, d_1\}$. 
\end{enumerate} 
\end{corollary}

\begin{proof}
Assume $h_i(x) = g_i(x)^{2 d_i}$ with $i \in \{0, 1\}$.
Since
\[
   h_i(x^s)^{s/d_i} = g_i((-1)^i)^{2 d_i \cdot s/d_i}
                    = g_i((-1)^i)^{q-1}
                    = 1
\]
for $x \in C_i$, we have $\log_{\xi} h_i((-1)^i) \equiv 0 \pmod{2 d_i}$.
If $d_0 \mid d_1$, then $\log_{\xi} h_0(1) \equiv \log_{\xi} h_1(-1) \pmod{2 d_0}$.
Thus $\phi(0, 2d_0) \ne \phi(1, 2d_0)$ \ifa $2 d_0 \nmid r_1$, namely, $r_1$ is odd.
Otherwise, $\phi(0, 2d_0) = \phi(1, 2d_0)$ \ifa $r_1$ is even.
Then the result follows from \cref{aixrimto1_2pieces}.
\end{proof}

Applying \cref{aixrimto1_gi(x)} to $r_0 = r_1$ 
yields the following result.

\begin{example}
Let $q$ be odd, $s = (q-1)/2$, and $g_0, g_1 \in \fqx$ with $g_0(1) g_1(-1) \ne 0$.
Let 
\[
 h(x) = g_0(x)^{2 d} (1 + x) + g_1(x)^{2 d} (1 - x),   
\]
where $r \in \n$ and $d = (r, s)$.
Then for $1 \le m < q$, 
$f(x) = x^r h(x^s)$ is \mfqstar \ifa 
(1) $m = d$ and $r$ is odd, or
(2) $m = 2d$ and $r$ is even.      
\end{example}


\begin{corollary}\label{aixrimto1_hi(x)}
Let $q$ be odd and $s = (q^n - 1)/2$, where $n \in \n$.
Let $r_i \in \n$, $d_i = (r_i, s)$, and $h_i \in \fqnx$ with $h_i((-1)^i) \in \fqstar$,
where $i \in \{0, 1\}$.
Suppose $2d \mid (q^n - 1)/(q - 1)$, where $d = \min \{d_0, d_1\}$. 
Then for $1 \le m < q^n$, 
\[
  f(x) = x^{r_0} h_0(x^s) (1 + x^s) + x^{r_1} h_1(x^s) (1 - x^s)
\]
is \mfield{m}{\fqnstar} \ifa one of the following holds:
\begin{enumerate}[\upshape(1)]
  \item $m = d_0 = d_1$ and $r_1$ is odd;
  \item $m = d_0 + d_1$, $d \mid m$,
        $r_1$ is even,
        and $s (m - 2d)/(m - d) < m$.
\end{enumerate} 
\end{corollary}

\begin{proof}
The result is a special case of \cref{aixrimto1_2pieces},
since $\log_{\xi} h_i((-1)^i) \equiv 0 \pmod{2 d}$ follows from
\[
    h_i((-1)^i)^{s/d} 
    = h_i((-1)^i)^{(q-1) \cdot (q^n - 1)/(2d(q - 1))}
    = 1.  \qedhere
\]
\end{proof}

For $h_k(x) = 1 + x + \cdots + x^{k-1}$, 
it is well known that $h_k(1) \ne 0$ \ifa $p \nmid k$,
and $h_k(-1) \ne 0$ \ifa $k$ is odd, 
where $p$ is the characteristic of $\fqn$.
Applying \cref{aixrimto1_hi(x)} to $h_i(x) = h_{k_i}(x^{e_i})^{t_i}$
and $r_0 = r_1$ yields the following result.

\begin{example}
Let $q$ be odd, $s = (q^n - 1)/2$, and $d = (r, s)$, where $n, r \in \n$.
Let $2d \mid (q-1, n)$, $p \nmid k_0$, and $k_1$ be odd,
where $p$ is the characteristic of $\fqn$.
Let $h_{k_i}(x) = 1 + x + \cdots + x^{k_i-1}$ and 
\[
  f(x) = x^{r} \big(h_{k_0}(x^{e_0})^{t_0} (1 + x) 
        + h_{k_1}(x^{e_1})^{t_1} (1 - x)\big) \circ x^s,
\]
where $e_i$ is odd and $t_i \in \n$ with $i \in \{0, 1\}$.
Then for $1 \le m < q^n$, $f$ is \mfield{m}{\fqnstar} \ifa 
(1) $m = d$ and $r$ is odd, or
(2) $m = 2d$ and $r$ is even.      
\end{example}

\subsection{Cyclotomic mappings} 

We next consider the special case
$r_0 = \cdots = r_{\ell - 1}$ and $h_0 = \cdots = h_{\ell - 1}$.
In this case, $f$ in \cref{eq:f=xrihixs} is reduced to 
the form $f(x) = x^r h(x^s)$.
Then we get the main theorem of \cite{Zhengmto1} by \cref{Lpiecesmto1}.

\begin{theorem}[{\cite[Thereom~4.3]{Zhengmto1}}]\label{mto1_xrh(xs)}
Let $q - 1 = \ell s$ and $d = (r, s)$, where $\ell, r, s \in \n$.
Let $f(x) = x^r h(x^s)$ and $g(x) = x^{r_1} h(x)^{s_1}$, 
where $r_1 = r / d$, $s_1 = s / d$, 
and $h \in \fqx$ has no roots in 
$U_\ell \coloneq \{\alpha \in \fqstar : \alpha^\ell = 1\}$.
Then for $1 \le m < q$, $f$ is \mfqstar \ifa $d \mid m$,
$g$ is \mfield{(m/d)}{U_\ell}, and $s (\ell \bmod{(m/d)}) < m$.
\end{theorem}

\begin{proof}
For the reason of completeness, we give a different proof.
For any $x \in C_i$, we have $x^s = \omega^i$, 
and so $f(x) = h(\omega^i) x^r$, where $\omega = \xi^s$. 
Note that $U_\ell = \{\omega^i : i \in [\ell]\}$, 
where $[\ell] = \{0, 1, \ldots, \ell-1\}$.
Since $h$ has no roots in $U_\ell$,
we have $h(\omega^i) \in \fqstar$ for any $i \in [\ell]$.
By \cref{Lpiecesmto1}, $f$ is \mfqstar \ifa 
$d \mid m$, $\phi(x, \ell d)$ is \mset{(m/d)}{[\ell]},
and $s(\ell \bmod{(m/d)}) < m$, where
$\phi(x, \ell d) = (r x + \log_{\xi} h(\omega^x)) \bmod{(\ell d)}$. 
Since $\omega = \xi^s$, we have $\omega^{r_1} = \xi^{s r_1} = \xi^{r s_1}$,
and so the following statements are equivalent:
\begin{enumerate}[(a)]
\item $\phi(x, \ell d)$ is \mset{(m/d)}{[\ell]};
\item $s_1(r x + \log_{\xi} h(\omega^x)) \bmod {(q-1)}$ is \mset{(m/d)}{[\ell]};
\item $(\xi^x)^{r s_1} h(\omega^x)^{s_1}$ is \mset{(m/d)}{[\ell]}; 
\item $(\omega^x)^{r_1} h(\omega^x)^{s_1}$ is \mset{(m/d)}{[\ell]};
\item $x^{r_1} h(x)^{s_1}$ is \mfield{(m/d)}{U_\ell}. \qedhere
\end{enumerate}
\end{proof}

Applying \cref{mto1_xrh(xs)} to $\ell = q+1$, 
$s = q-1$, and $d = 1$ yields the next result. 

\begin{corollary}\label{xrh(xs):fq2}
Let $f(x) = x^r h(x^{q-1})$ and $g(x) = x^{r} h(x)^{q-1}$, 
where $(r, q-1) = 1$ and $h \in \fqtwox$ has no roots in $U_{q+1}$.
Then for $1 \le m \le q+1$, $f$ is \mfield{m}{\fqtwostar} \ifa
$g$ is \mfield{m}{U_{q+1}} and $(q-1) (q+1 \bmod{m}) < m$.
\end{corollary}

This result gives us a recipe to construct \mtoone mappings on $\fqtwostar$ 
from \mtoone mappings on its subgroup $U_{q+1}$ under suitable conditions.  

Very recently, Hou and Lavorante \cite{Hou232193}  
presented a general construction that can
generate all permutation polynomials
of the form $x^r h(x^{q-1})$ over $\fqtwo$,
where $x^r h(x)^{q-1}$ induces monomial functions 
on the cosets of a subgroup of $U_{q+1}$. 
In the earlier sections, $f(x)$ in \cref{map_aixri} also induces
monomial functions on the cosets of a subgroup of $\fqstar$.
Thus, replacing $\fqstar$ by $U_{q+1}$ in the results of the previous sections, 
we can obtain the corresponding \mtoone mappings on $U_{q+1}$,
and then determine the \mtoone property of $x^r h(x^{q-1})$ on $\fqtwostar$ by \cref{xrh(xs):fq2}. 
We will review some results of \cite{Hou232193} in \cref{sec_pp},
and present the \mtoone versions of these results 
in \cref{sec_g_2,sec_g_m}.


\subsubsection{Permutations of Hou and Lavorante}
\label{sec_pp}

Let $\zeta$ be a generator of $U_{q+1}$ 
and $q+1 = \ell t$ for some $\ell$, $t \in \n$. 
Let $A_0 =  \{\zeta^{j \ell} : 0 \le j < t\}$ 
be the set of all $t$-th roots of unity in $U_{q+1}$.
The elements of the factor group $U_{q+1}/A_0$ 
are the cosets
\begin{equation}\label{eq:Ai}
    A_i = \zeta^{i} A_0 = \{\zeta^{j \ell + i} : 0 \le j < t\}, 
    \quad  0 \le i \le \ell - 1,
\end{equation}
which form a partition of $U_{q+1}$.
It was first observed in \cite{wangindex19} that permutation polynomials of the form 
$x+\text{Tr}_{q^2/q}(x^{(q+1)^2/4})$ over $\fqtwo$
can be induced by a generalized cyclotomic mapping $g$ of index $2$ over $U_{q+1}$. That is, $g$ behaves as a monomial function on each $A_0$ and $A_1$. 
An algorithm that produces all permutations 
of $\fqtwo$ of the form $x^r h(x^{q-1}) \in \fqtwox$
such that $x^rh(x)^{q-1}$ induces a monomial function 
on each $A_i$ in \cref{eq:Ai} 
was given by Hou and Lavorante \cite{Hou232193},
which unifies and generalizes some results in \cite{KLiQCL18,DZheng19,CaoHMX20,Lav2022,Qin2334}.
In particular, the expressions of all binomial and 
trinomial~$h(x)$ with such property were determined explicitly; 
see the following lemmas.


\begin{lemma}[\cite{Hou232193}]\label{lem_ppbi}
Let $q+1 = \ell t$ for some $\ell, t \in \n$.
Let $f(x) = x^r h(x^{q-1})$ permute $\fqtwo$, 
where $r \in \n$ and $h \in \fqtwox$ has no roots in $U_{q+1}$.
If $h$ is a binomial and $h^{q-1}$ induces a monomial 
function on $A_i$ for any $0 \le i \le \ell -1$,
then $f$ can only be one of the following forms:
\begin{enumerate}[\upshape(1)]
  \item \label{item_ppbi1} $f(x) = x^r(1 + a x^{v (q^2-1)/\ell})$, 
    where $a \in \fqtwostar$ and $0 \le v \le \ell-1$;
  \item \label{item_ppbi2} $f(x) = x^r(1 + a x^{(u + v t)(q-1)})$,
    where $a \in U_{q+1}$, $0 \le u \le t-1$, and $0 \le v \le \ell-1$;
  \item \label{item_ppbi3} $f(x) = x^r(1 + a x^{(1 + 2v)(q^2-1)/(2\ell)})$,
    where $a \in \fqtwostar$ and $0 \le v \le \ell-1$.
\end{enumerate} 
\end{lemma}

\begin{lemma}[{\cite[Table~2]{Hou232193}}]\label{lem_pptri}
Let $q+1 = \ell t$ for some $\ell, t \in \n$.
Let $f(x) = x^r h(x^{q-1})$ permute $\fqtwo$, 
where $r \in \n$ and $h \in \fqtwox$ has no roots in $U_{q+1}$.
If $h$ is a trinomial and $h^{q-1}$ induces a monomial 
function on $A_i$ for any $0 \le i \le \ell -1$,
then $f$ can only be one of the following forms:
\begin{enumerate}[\upshape(1)]
  \item \label{item_pptri1} $f(x) = x^r(1 + a x^{(q^2-1)/2} + b x^{(u + v t)(q-1)})$;
  \item \label{item_pptri2} $f(x) = x^r(1 - x^{(q^2-1)/2} + a x^{(u + v t)(q-1)})$;
  \item \label{item_pptri3} $f(x) = x^r (1 - x^{k t (q-1)} + a x^{(u + v t)(q-1)})$;
  \item \label{item_pptri4} $f(x) = x^r(1 + a x^{(q^2-1)/6} - a^{-1} x^{5(q^2-1)/6})$; 
  \item \label{item_pptri5} $f(x) = x^r(1 + a x^{5(q^2-1)/6} - a^{-1} x^{(q^2-1)/6})$. 
\end{enumerate} 
Here the conditions of the parameters are described in \cite[section~4.2]{Hou232193}.
\end{lemma}

Let $s = (q^2-1)/ \ell$ and $\omega = \xi^{s}$, where $\xi$ is a primitive element of $\fqtwo$.
For \cref{item_ppbi1} of \cref{lem_ppbi},  
\begin{equation*}
 f(x) = x^r(1 + a x^{s v}) 
      = (1 + a \omega^{iv}) x^r \quad \text{if $x \in C_i$, $0 \le i \le \ell -1$,}
\end{equation*}
and thus it is a special case of \cref{Lpiecesmto1}. 
Similarly, \cref{item_ppbi3} of \cref{lem_ppbi} and \cref{item_pptri4,item_pptri5} of \cref{lem_pptri} are also special cases of \cref{Lpiecesmto1}.

For the remaining polynomials in \cref{lem_ppbi,lem_pptri}, 
we will use \cref{xrh(xs):fq2} to convert the \mtoone property of $f(x) = x^r h(x^{q-1})$ on $\fqtwostar$ to the \mtoone property of $g(x) = x^r h(x)^{q-1}$ on $U_{q+1}$.
For $h \in \fqtwox$, 
if $h^{q-1}$ induces a monomial function on $A_i$ for any $0 \le i \le \ell -1$, 
then $g(x)$ also induces a monomial on $A_i$ for any $0 \le i \le \ell -1$. Assume 
\begin{equation}\label{map_lambdaixei}
  g(x) = \lambda_i x^{e_i}, \quad 
  \text{if $x \in A_i$, $0 \le i \le \ell - 1$},
\end{equation} 
where $\lambda_{i} \in U_{q+1}$ and $e_{i} \in \n$. 
When $\ell = 2$ or $\ell \ge 3$ and all $(e_i, t)$'s are equal, 
the \mtoone property of $g$ on $U_{q+1}$ can be determined
 by \cref{aixrimto1_2pieces,Lpiecesmto1}, respectively. 
In addition, we should replace $\phi(i, n)$ 
in \cref{aixrimto1_2pieces,Lpiecesmto1} by 
\[
\psi(i, n)= (i e_{i} + \log_{\zeta}\lambda_{i}) \bmod {n}.
\]

\subsubsection{The case \texorpdfstring{$g(x)$ has two branches on $U_{q+1}$}{g(x) has two branches on Uq+1}}\label{sec_g_2}

In this subsection, we consider the case~$g$ 
in \cref{map_lambdaixei} 
has two branches on $U_{q+1}$, i.e., 
\begin{equation}\label{map_lambdaixei_2}
  g(x) = 
  \begin{cases}
  \lambda_0 x^{e_0}  & \text{if $x \in A_0$,} \\
  \lambda_1 x^{e_1}  & \text{if $x \in A_1$.}
  \end{cases}   
\end{equation}

Replacing $\fqstar$ by $U_{q+1}$ and $\phi(i, n)$ by $\psi(i, n)$ 
in \cref{aixrimto1_2pieces},
we can characterize the \mtoone property of~$g$ on $U_{q+1}$,
and then we obtain the next result by \cref{xrh(xs):fq2}.

\begin{theorem}\label{2piecesmto1_Fq2}
Let $q$ be odd and $t = (q+1)/2$.
Let $f(x) = x^r h(x^{q-1})$ and $g(x) = x^{r} h(x)^{q-1}$, 
where $(r, q-1) = 1$ and $h \in \fqtwox$ has no roots in $U_{q+1}$.
Assume $g$ can be rewritten in the form \cref{map_lambdaixei_2}
and $d_i = (e_i, t)$ with $i \in \{0, 1\}$.
Then for $1 \le m \le q + 1$, $f$ is \mfield{m}{\fqtwostar} \ifa 
$(q-1) (q+1 \bmod{m}) < m$ and one of the following holds:
\begin{enumerate}[\upshape(1)]
  \item \label{item_e0t=m} $m = d_0 = d_1$ and $\psi(0, 2m) \ne \psi(1, 2m)$;
  \item \label{item_e0t<m} $m = d_0 + d_1$, $d \mid m$, $\psi(0, 2d) = \psi(1, 2d)$,
        and $t (m - 2d)/(m - d) < m$, where $d = \min \{d_0, d_1\}$.        
\end{enumerate} 
\end{theorem}


Applying \cref{2piecesmto1_Fq2} to the~$f$ 
in \cref{item_ppbi2} of \cref{lem_ppbi} 
with $v = 1$ yields the next result.


\begin{corollary}\label{co_2piecemto1}
Let $q$ be odd and $t = (q+1)/2$.
Let $f(x) = x^r (1 + a x^{(u+t)(q-1)})$,
where $(r, q-1) = 1$, $0 \le u < t$, and $a \in U_{q+1}$ 
with $(-a)^{(q+1)/(q+1, u+t)} \ne 1$.
Then for $1 \le m \le q + 1$, $f$ is \mfield{m}{\fqtwostar} 
\ifa one of the following holds:
\begin{enumerate}[\upshape(1)]
  \item $m = (r - u, t) $ and $r - u \not \equiv t \pmod{2m}$;
  \item $m = 2(r - u, t)$ and $r - u \equiv t \pmod{m}$.
\end{enumerate} 
\end{corollary}
\begin{proof}
Let $h(x) = 1 + a x^{u+t}$ with $a \in U_{q+1}$. 
Then $f(x) = x^r h(x^{q-1})$ and $h$ has no roots 
in $U_{q+1}$ \ifa $(-a)^{(q+1)/(q+1, u+t)} \ne 1$. 
For any $x \in A_0$, we have $x^{t} = 1$ and so
\begin{align*}
     g(x) \coloneq x^r h(x)^{q-1} 
    = x^r \frac{(1 + a x^{u})^q}{1 + a x^{u}} 
    = x^r \frac{1 + a^{-1} x^{-u}}{1 + a x^{u}}
    = a^{-1} x^{r-u}.  
  \end{align*}
Similarly, for any $x \in A_1$, we get $x^{t} = - 1$ and so $g(x) = -a^{-1} x^{r-u}$. By \cref{2piecesmto1_Fq2}, for $1 \le m \le q + 1$, $f$ is \mfield{m}{\fqtwostar} \ifa $(q-1) (q+1 \bmod{m}) < m$ and one of the following holds:
\begin{enumerate}[\upshape(1)]
  \item $m = (r - u, t)$ and 
        $\log_{\zeta} {a^{-1}} \not \equiv r - u + \log_{\zeta} {(-a^{-1})} \pmod{2m}$;
  \item $m = 2 (r - u, t)$ and 
        $\log_{\zeta} {a^{-1}} \equiv r - u + \log_{\zeta} {(-a^{-1})} \pmod{m}$.
\end{enumerate}
Obviously, $(q-1) (q+1 \bmod{m}) < m$ and $\log_{\zeta} {a^{-1}} - \log_{\zeta} {(-a^{-1})} = \log_{\zeta} {(-1)} = t$.
\end{proof}

When $v=0$ in \cref{item_ppbi2} of \cref{lem_ppbi}, 
a similar argument gives the following result.

\begin{corollary}
Let $q$ be odd and $t = (q+1)/2$.
Let $f(x) = x^r (1 + a x^{u(q-1)})$,
where $(r, q-1) = 1$, $0 \le u < t$, and $a \in U_{q+1}$ with $(-a)^{(q+1)/(q+1, u)} \ne 1$.
Then for $1 \le m \le q + 1$, $f$ is \mfield{m}{\fqtwostar} \ifa $(r-u, q+1) = m$.
\end{corollary}


Applying \cref{2piecesmto1_Fq2} to the~$f$ 
in \cref{item_pptri1} of \cref{lem_pptri} 
yields the next result.
 
\begin{corollary}\label{cor_tri1ab}
Let $q$ be odd, $t = (q+1)/2$, $0 < u < t$, and $v \in \{0, 1\}$.
Let $a, b \in \fqtwostar$ satisfy $a^{q-1} = -1$ and $b^{q+1} = 1 - a^2$.
Let
\[
    f(x) = x^r (1 + a x^{t(q-1)} + b x^{(u + v t)(q-1)}),
\] 
where $(r, q-1) = 1$ and $1 + a x^t + b x^{u + v t}$ has no roots in $U_{q+1}$. 
Let $w = \log_{\zeta} {(-1)^v (1-a)/(1+a)}$.
Then for $1 \le m \le q + 1$, $f$ is \mfield{m}{\fqtwostar} \ifa one of the following holds:
\begin{enumerate}[\upshape(1)]
  \item $m = (r - u, t)$ and $r - u \not \equiv w \pmod{2m}$;
  \item $m = 2 (r - u, t)$ and $r - u \equiv w \pmod{m}$.  
\end{enumerate} 
\end{corollary}

\begin{proof}
Let $h(x) = 1 + a x^t + b x^{u + v t}$. 
Then $f(x) = x^r h(x^{q-1})$ and for any $x \in U_{q+1}$,
\begin{equation*}
     g(x) \coloneq x^r h(x)^{q-1}
        = x^r \frac{(1 + a x^t + b x^{u + v t})^q}{1 + a x^t + b x^{u + v t}}
        = x^r \frac{1 + a^q x^{-t} + b^q x^{- u - v t}}{1 + a x^t + b x^{u + v t}}.
\end{equation*}
Since $a^{q-1} = -1$ and $b^{q+1} = 1 - a^2$, we have $a^q = -a$ and $b^q = (1- a^2)b^{-1}$, and thus
\begin{equation*}
    g(x) = x^r \frac{1 -a x^{-t} + (1-a^2)b^{-1} x^{- u - v t}}{1 + a x^t + b x^{u + v t}}.
\end{equation*}
Because $x^t = 1$ for $x \in A_0$ and $x^t = -1$ for $x \in A_1$, we get
\begin{equation*}
  g(x) = 
  \begin{cases}
  (1-a)b^{-1} x^{r-u}  & \text{if $x \in A_0$,} \\
  (-1)^v (1 + a) b^{-1} x^{r-u}  & \text{if $x \in A_1$.}
  \end{cases}   
\end{equation*}
The condition $b^{q+1} = 1 - a^2 = (1 \pm a)^{q+1}$ implies 
$(1 \pm a) b^{-1} \in U_{q+1}$.
Then the result follows from \cref{xrh(xs):fq2} and \cref{2piecesmto1_Fq2}.
\end{proof}

In \cref{cor_tri1ab}, take $q \equiv 3 \pmod{4}$, $u = 1$, $v = 0$, $b = -(1-a)$, and $b^t = -(1+a)^t$.
Then $t$ is even and $w$ is odd. 
Moreover, since $(r, q-1) = 1$, 
we get $r$ is odd and so $r - 1$ is even.
Thus $(r-1, t) \ne 1$ and $w \not \equiv r - 1 \pmod{4}$. 
Then we obtain the next example.
\begin{example}
Let $q \equiv 3 \pmod{4}$ and $t = (q+1)/2$.
Let $a \in \fqtwostar$ satisfy $a^{q-1} = -1$ and $(1-a)^t = -(1+a)^t$.
Let
\[
    f(x) = x^r (1 + a x^{t(q-1)} - (1-a) x^{q-1}), 
\]
where $(r, q-1) = 1$.
Then $f$ is \mfield{2}{\fqtwostar} \ifa $(r - 1, t) = 2$.
\end{example}

The existence of $h(x) := 1 + a x^t - (1 - a) x$ 
such that $h$ has no roots in $U_{q+1}$ 
is explained in \cite[Example~4.12]{Hou232193}.

Applying \cref{2piecesmto1_Fq2} to the~$f$ 
in \cref{item_pptri2} of \cref{lem_pptri} 
yields the next result.

\begin{corollary}\label{cor_tri1a}
Let $q$ be odd, $t = (q+1)/2$, $0 < u < t$, 
and $v \in \{0, 1\}$.
Let $a \in \fqtwostar$ satisfy $a^{q+1} = 4$. Let
\[
    f(x) = x^r (1 - x^{t(q-1)} + a x^{(u + v t)(q-1)}),
\] 
where $(r, q-1) = 1$ and $1 - x^t + a x^{u + v t}$ 
has no roots in $U_{q+1}$. 
Then for $1 \le m \le q + 1$, $f$ is 
\mfield{m}{\fqtwostar} \ifa $(q-1) (q+1 \bmod{m}) < m$ and
one of the following holds:
\begin{enumerate}[\upshape(1)]
  \item $m = (r - u, t) = (r - 2 u, t)$ and 
        $\log_{\zeta} {(-1)^v 2 a^{-1}} \not \equiv r - u \pmod{2m}$;
  \item $m = (r - u, t) + (r - 2 u, t)$, $d \mid m$, 
        $\log_{\zeta} {(-1)^v 2 a^{-1}} \equiv r - u \pmod{2d}$,
        and $t (m - 2d)/(m - d) < m$, 
        where $d = \min \{(r - u, t), (r - 2 u, t)\}$.
\end{enumerate} 
\end{corollary}

\begin{proof}
Let $h(x) = 1 - x^t + a x^{u + v t}$ 
and $g(x) = x^r h(x)^{q-1}$.
For $x \in A_0$, $x^t = 1$ and so $h(x) = a x^u$. 
Then $g(x) = a^{q-1} x^{r-2u}$.
For $x \in A_1$, we have $x^t = -1$ and so 
$h(x) = 2 + (-1)^v a x^u$. Then 
\[
    g(x) = x^r \frac{h(x)^q}{h(x)}
         = x^r \frac{2 + (-1)^v a^q x^{-u}}{2 + (-1)^v a x^u}
         = x^r \frac{2 + (-1)^v 4 a^{-1} x^{-u}}{2 + (-1)^v a x^u}
         = (-1)^v 2 a^{-1}x^{r-u}.
\]
By $a^{q+1} = 4 = 2^{q+1}$, we get $2 a^{-1} \in U_{q+1}$. 
Then the result follows from \cref{xrh(xs):fq2} and \cref{2piecesmto1_Fq2}.
\end{proof}
An argument similar to the one used in \cite[Section~4.4]{Hou232193} shows that \cite[Theorem~2.1]{Qin2334}, \cite[Theorem~5]{GoharZ16}, and \cite[Theorem~4.5]{Qin2469} are special cases of \cref{cor_tri1a}.

In \cref{cor_tri1a}, take $q \equiv 3 \pmod{4}$, $v = 0$, and $(2/a)^t = -1$. 
Then $h(x) = 1 - x^t + a x^u$. An argument similar to the one used in \cite[Example~4.13]{Hou232193} shows that $h$ has no roots in $U_{q+1}$. Since $r-u$ is odd, we have $(r-u, t) \ne 2$ and $\log_{\zeta} {2 a^{-1}} \equiv 1 \equiv r - u \pmod{2}$. 
Then we get the next example.
\begin{example}
Let $q \equiv 3 \pmod{4}$ and $t = (q+1)/2$.
Let $a \in \fqtwostar$ satisfy $(2/a)^t = -1$.
Let
\[
    f(x) = x^r (1 - x^{t(q-1)} + a x^{u(q-1)}),
\]
where $(r, q-1) = 1$ and $u$ is even. 
Then $f$ is \mfield{2}{\fqtwostar} \ifa $(r - u, t) = (r - 2 u, t) = 1$.
\end{example}

\subsubsection{The case \texorpdfstring{$g(x)$ has multiple branches on $U_{q+1}$}{g(x) has multiple branches on Uq+1}}\label{sec_g_m}

This subsection is devoted to the case~$g$ in 
\cref{map_lambdaixei} has multiple branches on $U_{q+1}$. 
Replacing $\fqstar$ by $U_{q+1}$ and 
$\phi(i, n)$ by $\psi(i, n)$ in \cref{Lpiecesmto1},
we can characterize the \mtoone property of~$g$ on $U_{q+1}$,
and then we obtain the following result by \cref{xrh(xs):fq2}.

\begin{theorem}\label{Lpiecesmto1_Fq2}
Let $q+1 = \ell t$ and $[\ell] = \{0, 1, \ldots, \ell-1\}$.
Let $f(x) = x^r h(x^{q-1})$ and $g(x) = x^{r} h(x)^{q-1}$, where $(r, q-1) = 1$ and $h \in \fqtwox$ has no roots in $U_{q+1}$. 
Assume $g$ can be rewritten in the form \cref{map_lambdaixei},
and $(e_i, t) = d$ for any $i \in [\ell]$.
Then for $1 \le m \le q+1$, $f$ is \mfield{m}{\fqtwostar} \ifa 
$d \mid m$, $\psi(x, \ell d)$ is \mset{(m/d)}{[\ell]}, $t(\ell \bmod{(m/d)}) < m$, and $(q-1) (q+1 \bmod{m}) < m$.
\end{theorem}

Applying \cref{Lpiecesmto1_Fq2} to the~$f$ 
in \cref{item_pptri3} of \cref{lem_pptri} 
with $\ell = 3$ yields the next result.

\begin{corollary}\label{cor_trikuv}
Let $q \equiv 2 \pmod{3}$ with $q \ge 5$, $t = (q+1)/3$, $0 < u < t$, $0 \le v < 2$ and $k \in \{1, 2\}$.
Let $a \in \fqtwostar$ satisfy $(1 - \epsilon^k)/a \in U_{q+1}$, where $\epsilon = \zeta^t$.
Let
\[
    f(x) = x^r (1 - x^{k t (q-1)} + a x^{(u + v t)(q-1)}),
\] 
where $(r, q-1) = 1$ and $1 - x^{k t} + a x^{u + v t}$ has no roots in $U_{q+1}$. 
Assume $(r-u, t) = (r-2u, t) = d$. 
Let $z_1 = \log_{\zeta} a^{-1}(1 - \epsilon^{k}) \epsilon^{v}$ and $z_2 = \log_{\zeta} a^{-1}(1 - \epsilon^{-k}) \epsilon^{-v}$.
Then for $1 \le m \le q+1$, $f$ is \mfield{m}{\fqtwostar} \ifa one of the following holds:
\begin{enumerate}[\upshape(1)]
  \item \label{item_m=d}$m = d$, $i(r - u) \not \equiv z_i \pmod{3d}$ for any $i \in \{1, 2\}$, 
        and $r - u \not \equiv z_2 - z_1 \pmod{3d}$;
  \item \label{item_m=3d}$m = 3d$ and $i(r - u) \equiv z_i \pmod{3d}$ for any $i \in \{1, 2\}$. 
\end{enumerate} 
\end{corollary}

\begin{proof}
Let $h(x) = 1 - x^{k t} + a x^{u + v t}$ and $g(x) = x^r h(x)^{q-1}$. 
For any $x \in A_0$, $x^t = 1$ and so $h(x) = a x^u$. 
Then $g(x) = a^{q-1} x^{r-2u}$.
For $x \in A_1$, we get $x^t = \epsilon$ and so 
$h(x) = 1 - \epsilon^k + a \epsilon^v x^u$.
Since $(1 - \epsilon^k)/a \in U_{q+1}$, we have 
$a^{q+1} = (1 - \epsilon^k)^{q+1} = (1-\epsilon^{-k})(1 - \epsilon^k)$ and thus 
\begin{align*}
g(x) & = x^r \frac{h(x)^q}{h(x)}
    = x^r \frac{1 - \epsilon^{-k} + a^q \epsilon^{-v} x^{-u}}{1 - \epsilon^k + a \epsilon^v x^u}\\
    & = (1 - \epsilon^{-k}) x^r 
      \frac{1 + a^{-1}(1 - \epsilon^k) \epsilon^{-v} x^{-u}}{1 - \epsilon^k + a \epsilon^v x^u} \\
    & = a^{-1}(1 - \epsilon^{-k}) \epsilon^{-v} x^{r-u}. 
\end{align*}
For $x \in A_2$, $x^t = \epsilon^{-1}$ and so 
$h(x) = 1 - \epsilon^{-k} + a \epsilon^{-v} x^u$.
Similarly, $g(x) = a^{-1}(1 - \epsilon^{k}) \epsilon^{v} x^{r-u}$.

Obviously, $f(x) = x^r h(x)^{q-1}$ and $a^{q-1}$,
$a^{-1}(1 - \epsilon^{k}) \epsilon^{v} \in U_{q+1}$.
Since $(1 - \epsilon^k)^{q+1} = (1 - \epsilon^{-k})^{q+1}$, we have 
$a^{-1}(1 - \epsilon^{-k}) \epsilon^{-v} \in U_{q+1}$.
By \cref{Lpiecesmto1_Fq2}, if $f$ is \mfield{m}{\fqtwostar},
then $\psi(x, 3 d)$ is \mset{(m/d)}{[3]} and thus $1 \le m/d \le 3$.
If $m = d$, then $f$ is \mfield{m}{\fqtwostar} \ifa $\psi(x, 3 d)$ is \mset{1}{[3]}.
That is, \cref{item_m=d} holds.
If $m = 2d$, then $t < m = 2d$, i.e., $d = t$. 
Since $q \ge 5$, we have $(q-1) (q+1 \bmod{2t}) = (q-1)t > 2t = m$ and thus $f$ is not \mfield{m}{\fqtwostar}.
If $m = 3d$, then $f$ is \mfield{m}{\fqtwostar} \ifa $\psi(x, 3 d)$ is \mset{3}{[3]}.
That is, \cref{item_m=3d} holds.
\end{proof}

In \cref{cor_trikuv}, take $q = 2^5$, $r = 6$, $k = 1$, 
$u = 2$, $v = 0$, $a = \zeta^{-1} (1 + \epsilon)$. 
Then $(r, q-1) = 1$,  $t = 11$ and so $ (r-u, t) = (r-2u, t) = 1$.
Moreover, $z_1 = 1$, $z_2 = 1-t$, and $h(x) = 1 + x^t + \zeta^{-1} (1 + \epsilon) x^2$. 
An argument similar to the one used in \cite[Example~4.14]{Hou232193} shows that $h$ has no roots in $U_{q+1}$. 
Since $i(r - u) \equiv z_i \pmod{3}$ for $i \in \{1, 2\}$, we obtain the next example.
\begin{example}
Let $q = 2^5$ and $\zeta$ be a generator of $U_{q+1}$. Then
\[
    f(x) = x^6 (1 + x^{11(q-1)} + (\zeta^{-1} +\zeta^{10}) x^{2(q-1)})
\]
is \mfield{3}{\fqtwostar}.
\end{example}

\section{Conclusions}
In this paper, we completely characterize the many-to-one property 
of generalized cyclotomic mappings for index $\ell \leq 3$.
For any divisor $\ell$ of $q-1$, 
we completely determine their \twotoone property,
and characterize their many-to-one property 
under the condition that all $(r_i, s)$'s are equal.
Moreover,  explicit expressions of some \mtoone mappings on $\fqstar$ are obtained 
by using special polynomials to represent the constants $a_i$'s in their cyclotomic representations. 
This includes several classes of many-to-one binomials 
and trinomials of the form $x^r h(x^{q-1})$ over $\fqtwostar$.

\vspace{12pt}

\noindent
\textbf{Data availability} No datasets were generated or analysed during the current study.

\noindent
\textbf{Competing interests} The authors declare no competing interests.

\bibliography{jrnlabbr.bib, mto1.bib} 

\end{document}